\documentclass[11pt,a4paper]{article}

\usepackage{geometry}
\usepackage{vkmacros}

\newcommand{\R}{\mathbb R}

\newcommand{\eps}{\varepsilon}

\newcommand{\de}{{\rm d}}

\newcommand{\Rd}{\mathbb R_X\!\left(d\right)}
 \newcommand{\Rslb}{\ushort{\mathbb R}_X\!\left(d\right)}

\newgeometry{tmargin=2.5cm, bmargin=2.5cm, lmargin=2.5cm, rmargin=2.5cm}

\begin{document}
\large \title{A lower bound on the differential entropy of log-concave random vectors with applications}

\author{Arnaud Marsiglietti$^1$\thanks{Supported by the Walter S. Baer and Jeri Weiss CMI Postdoctoral Fellowship.} ~and Victoria Kostina$^2$\thanks{Supported in part by the National Science Foundation (NSF)
under Grant CCF-1566567.} \\[10pt] California Institute of Technology \\ Pasadena, CA 91125, USA \\ $^1$amarsigl@caltech.edu \qquad $^2$vkostina@caltech.edu
}

\date{}

\maketitle

\begin{abstract}

We derive a lower bound on the differential entropy of a log-concave random variable $X$ in terms of the $p$-th absolute moment of $X$. The new bound leads to a reverse entropy power inequality with an explicit constant, and to new bounds on the rate-distortion function and the channel capacity.

Specifically, we study the rate-distortion function for log-concave sources and distortion measure $| x - \hat x|^r$, and we establish that the difference between the rate distortion function and the Shannon lower bound is at most $\log(\sqrt{\pi e}) \approx 1.5$ bits, independently of $r$ and the target distortion $d$. For mean-square error distortion, the difference is at most $\log (\sqrt{\frac{\pi e}{2}}) \approx 1$ bits, regardless of $d$.

We also provide bounds on the capacity of memoryless additive noise channels when the noise is log-concave. We show that the difference between the capacity of such channels and the capacity of the Gaussian channel with the same noise power is at most $\log (\sqrt{\frac{\pi e}{2}}) \approx 1$ bits.

Our results generalize to the case of vector $X$ with possibly dependent coordinates, and to $\gamma$-concave random variables. Our proof technique leverages tools from convex geometry.

\end{abstract}

\vspace*{0.5cm}


\section{Introduction}
\label{sec:intro}

It is well known that the differential entropy among all zero-mean random variables with the same second moment is maximized by the Gaussian distribution:
\begin{equation}
 h(X) \leq \log(\sqrt {2 \pi e \EE[|X|^2]}). \label{h2}
\end{equation}
More generally, the differential entropy under $p$-th moment constraint is upper bounded as (see e.g. \cite[Appendix 2]{ZF}), for $p>0$,
\begin{equation}
\begin{aligned}
h(X) \leq \log \left( \alpha_p \|X\|_p \right), \label{hp}
\end{aligned}
\end{equation}
where
\begin{align}
 \alpha_p &\triangleq 2 e^{\frac 1 p} \Gamma \left( 1 + \frac{1}{p} \right) p^{\frac 1 p}, \label{alphar} \\
 \|X\|_p &\triangleq \left(\EE[|X|^p]\right)^{\frac 1 p}.
\end{align}
Here, $\Gamma$ denotes the Gamma function. Of course, if $p = 2$, $\alpha_p = \sqrt{2 \pi e}$, and \eqref{hp} reduces to \eqref{h2}. 
A natural question to ask is whether a matching lower bound on $h(X)$ can be found in terms of $p$-norm of $X$, $\|X\|_p$.  The quest is meaningless without additional assumptions on the density of $X$, as $h(X) = -\infty$ is possible even if $\|X\|_p$ is finite. In this paper, we show that if the density of $X$, $f_X(x)$, is \emph{log-concave} (that is,  $\log f_X(x)$ is concave), then $h(X)$ stays within a constant from the upper bound in \eqref{hp} (see \thmref{variable-gen} in Section \ref{sec:main} below):
\begin{eqnarray}
h(X) \geq \log \frac{2 \|X - \EE[X]\|_p }{\Gamma(p+1)^{\frac 1 p}}, \label{eq:var-gen}
\end{eqnarray}
where $p \geq 1$. Moreover, the bound \eqref{eq:var-gen} tightens for $p=2$, where we have
\begin{eqnarray}
h(X) \geq \frac{1}{2} \log (4 \Var{X}). \label{eq:var-gen-2}
\end{eqnarray}

The bound \eqref{eq:var-gen} actually holds for $p>-1$ if, in addition to being log-concave, $X$ is \emph{symmetric} (that is, $f_X(x) = f_X(-x)$),  (see \thmref{variable} in Section \ref{sec:main} below).

The class of log-concave distributions is rich and contains important distributions in probability, statistics and analysis. Gaussian distribution, Laplace distribution, uniform distribution on a convex set, chi distribution are all log-concave. The class of log-concave random vectors has good behavior under natural probabilistic operations: namely, a famous result of Pr\'ekopa \cite{P} states that sums of log-concave random vectors, as well as marginals of log-concave random vectors, are log-concave. Furthermore, log-concave distributions have moments of all orders.

Together with the classical bound in \eqref{hp}, the bound in \eqref{eq:var-gen} tells us that entropy and moments of log-concave random variables are comparable.

Using a different proof technique, Bobkov and Madiman \cite{BM} recently showed that the differential entropy of a log-concave $X$ satisfies,
\begin{equation}
h(X) \geq \frac{1}{2} \log\left( \frac 1 2 \Var X \right). \label{eq:BM}
\end{equation}
Our results in \eqref{eq:var-gen} and \eqref{eq:var-gen-2} tighten \eqref{eq:BM}, in addition to providing a comparison with other moments.

Furthermore, this paper generalizes the lower bound on differential entropy in \eqref{eq:var-gen} to random vectors. If the random vector $X = (X_1, \ldots, X_n)$ consists of independent random variables, then the differential entropy of $X$ is equal to the sum of differential entropies of the component random variables, one can trivially apply \eqref{eq:var-gen} component-wise to obtain a lower bound on $h(X)$. In this paper, we show that even for nonindependent components, as long as the density of the random vector $X$ is log-concave and satisfies a symmetry condition, its differential entropy is bounded from below in terms of covariance matrix of $X$ (see \thmref{vector} in \secref{sec:main} below). As noted in \cite{BM-hyperplane}, such a generalization is related to the famous hyperplane conjecture in convex geometry. We also extend our results to a more general class of random variables, namely, the class of $\gamma$-concave random variables, with $\gamma < 0$. 

The bound \eqref{eq:var-gen} on the differential entropy allow us to derive reverse entropy power inequalities with explicit constants. The fundamental entropy power inequality of Shannon \cite{shannon1948mathematical} and Stam \cite{stam1959some} states that for every independent continuous random vector $X$ and $Y$ in $\R^n$,
\begin{equation}\label{EPI}
N(X+Y) \geq N(X) + N(Y),
\end{equation}
where
\begin{equation}\label{entropy-power}
N(X) = e^{\frac{2}{n} h(X)}
\end{equation}
denotes the entropy power of $X$. It is of interest to characterize distributions for which a reverse form of \eqref{EPI} holds. In this direction, it was shown by Bobkov and Madiman \cite{BM2} that given any continuous log-concave random vectors $X$ and $Y$ in $\R^n$, there exist affine volume-preserving maps $u_1, u_2$ such that a reverse entropy power inequality holds for $u_1(X)$ and $u_2(Y)$:
\begin{equation}\label{reverse-EPI-BM}
N(u_1(X) + u_2(Y)) \leq c(N(u_1(X)) + N(u_2(Y))) = c(N(X) + N(Y)),
\end{equation}
for some universal constant $c \geq 1$ (independent of the dimension).

In applications, it is important to know the precise value of the constant $c$ that appears in \eqref{reverse-EPI-BM}. It was shown by Cover and Zhang \cite{CZ} that if $X$ and $Y$ are identically distributed (possibly dependent) log-concave random variables, then
\begin{equation}\label{eq:rev-EPI-iid}
N(X+Y) \leq 4 N(X).  
\end{equation}
Inequality \eqref{eq:rev-EPI-iid} easily extends to random vectors (see \cite{MK}). A similar bound for difference of i.i.d. log-concave random vectors was obtained in \cite{BM3}, and reads as
\begin{equation}
N(X-Y) \leq e^2 N(X).  
\end{equation}
Recently, a new form of reverse entropy power inequality was investigated in \cite{BNT}, and a general reverse entropy power-type inequality was developed in \cite{C}. For further details, we refer to the survey paper \cite{MMX}. In \secref{reverse-EPI}, we provide explicit constants for non-identically distributed and uncorrelated log-concave random vectors (possibly dependent). In particular, we prove that as long as log-concave random variables $X$ and $Y$ are uncorrelated,
\begin{equation}\label{eq:rev-EPI}
N(X+Y) \leq \frac{\pi e}{2} (N(X) + N(Y)).
\end{equation}
A generalization of \eqref{eq:rev-EPI} to arbitrary dimension is stated in \thmref{reverse-vector} in Section \ref{sec:main} below.

The bound \eqref{eq:var-gen} on the differential entropy is essential in the study of the difference between the rate-distortion function and the Shannon lower bound that we describe next. 
Given a nonnegative number $d$, the rate-distortion function $\Rd$ under $r$-th moment distortion measure is defined as
\begin{eqnarray}
\Rd = \inf_{\substack{ P_{\hat X | {X}} \colon\\ \EE[|X-\hat{X}|^r] \leq d }} I(X;\hat{X}), \label{eq:Rd}
\end{eqnarray}
where the infimum is over all transition probability kernels $\mathbb R \mapsto \mathbb R$ satisfying the moment constraint. The celebrated Shannon lower bound \cite{shannon1959coding} states that the rate-distortion function is lower bounded by
\begin{align}
\Rd &\geq \Rslb \\
&\triangleq h(X) -\log \left( \alpha_r d^{\frac 1 r} \right), \label{eq:slb}
\end{align}
where $\alpha_r$ is defined in \eqref{alphar}. 
For mean-square distortion ($r = 2$), \eqref{eq:slb} simplifies as
\begin{eqnarray}
\Rd \geq h(X) - \log\sqrt{2\pi e d}.  \label{eq:slb2}
\end{eqnarray}
The Shannon lower bound states that the rate-distortion function is lower bounded by the difference between the differential entropy of the source and the term that increases with target distortion $d$, explicitly linking the storage requirements for $X$ to the information content of $X$ (measured by $h(X)$) and the desired reproduction distortion $d$. 
As shown in \cite{linkov1965evaluation,linder1994asymptotic,koch2015shannonlb} under progressively less stringent assumptions,\footnote{Koch \cite{koch2015shannonlb} showed that \eqref{eq:linkov} holds as long as $H(\lfloor X \rfloor ) < \infty$.} the Shannon lower bound is tight in the limit of low distortion,
\begin{eqnarray}
0 \leq \Rd - \Rslb \xrightarrow[d \to 0]{} 0. \label{eq:linkov}
\end{eqnarray}
The speed of convergence in \eqref{eq:linkov} and its finite blocklength refinement were recently explored in \cite{K1,K2}. Due to its simplicity and tightness in the high resolution / low distortion limit, the Shannon lower bound can serve as a proxy for the rate-distortion function $\Rd$, which rarely has an explicit representation. Furthermore, the tightness of the Shannon lower bound at low $d$ is linked to the optimality of simple lattice quantizers \cite{K1,K2}, an insight which has evident practical significance.
Gish and Pierce \cite{gish1968asymptotically} showed that for mean-square error distortion, the difference between the entropy rate of a dithered scalar quantizer, $H_1$, and the rate-distortion function $\Rd$ converges to $\frac 1 2 \log \frac {2 \pi e}{12} \approx 0.254$ bit/sample in the limit $d \downarrow 0$. Ziv \cite{ziv1985universal} proved that $H_1 - \Rd$ is bounded by $\frac 1 2 \log \frac {2 \pi e} 6 \approx 0.754$ bit/sample, universally in $d$.

In this paper, we show that the gap between $\Rd$ and $\Rslb$ is bounded universally in $d$, provided that the source density is log-concave: for mean-square error distortion ($r = 2$ in \eqref{eq:Rd}), we have
\begin{equation}
\Rd - \Rslb \leq \log \sqrt{\frac{\pi e}{2}} \approx 1.05 \mbox{ bits}. 
\end{equation}
\figref{fig:slb-gen} presents our bound for different values of $r$. Regardless of $r$ and $d$, 
\begin{equation}
 \Rd - \Rslb \leq \log(\sqrt{\pi e}) \approx 1.55 \mbox{ bits}. 
\end{equation}

\begin{figure}[htp]
\centerline{\includegraphics[width=8cm,height=5cm]{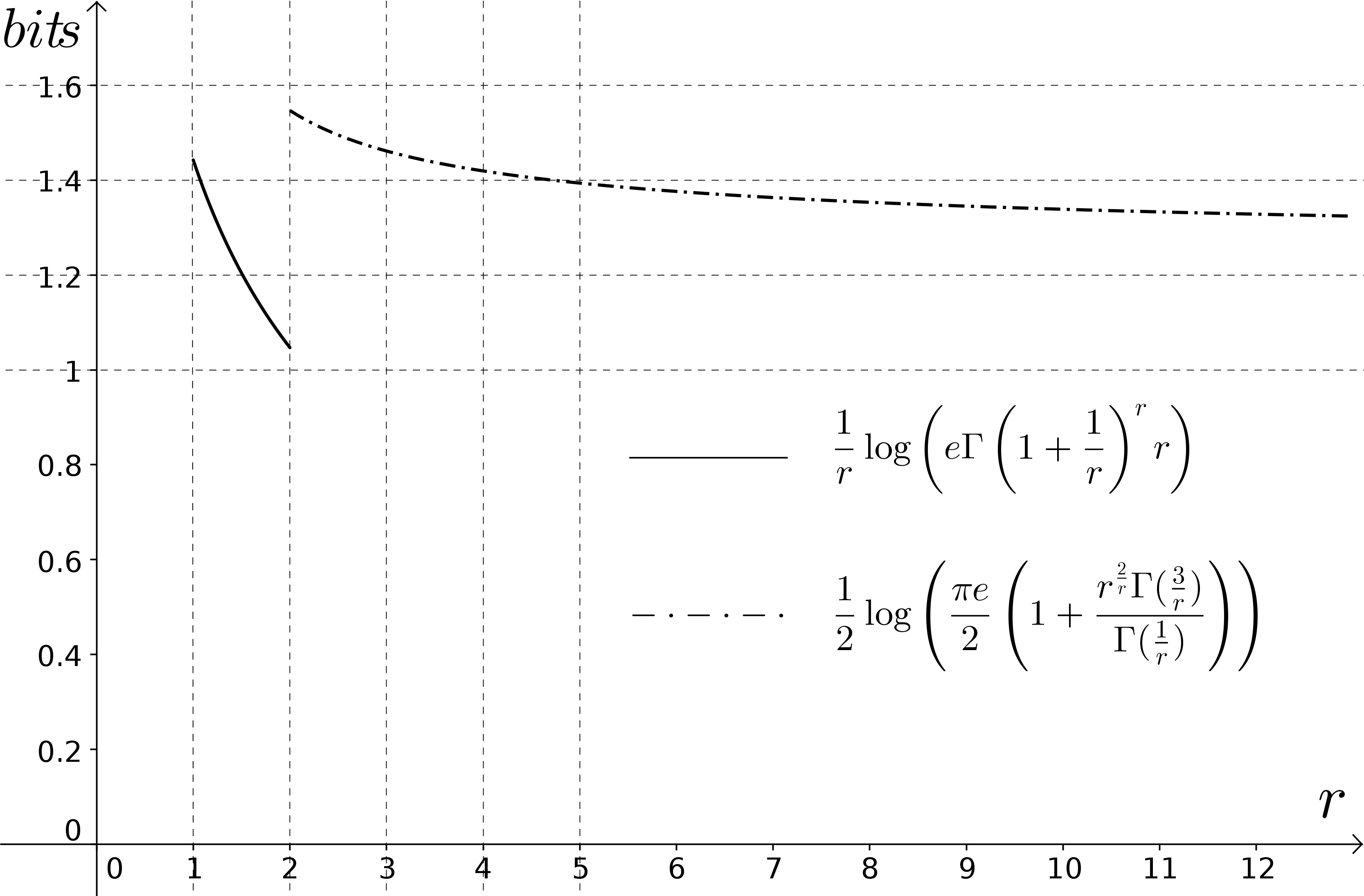}}
\caption{The bound on the difference between the rate-distortion function under $r$-th moment constraint and the Shannon lower bound, stated in Corollary \ref{cons-gen}.}
\label{fig:slb-gen}
\end{figure}

The bounds in \figref{fig:slb-gen} tighten for symmetric log-concave sources when $r \in (2;4.3)$. \figref{fig:slb} presents this tighter bound for different values of $r$. Regardless of $r$ and $d$, 
\begin{equation}
 \Rd - \Rslb \leq \log(e) \approx 1.44 \mbox{ bits}. 
\end{equation}

\begin{figure}[htp]
\centerline{\includegraphics[width=8cm,height=5cm]{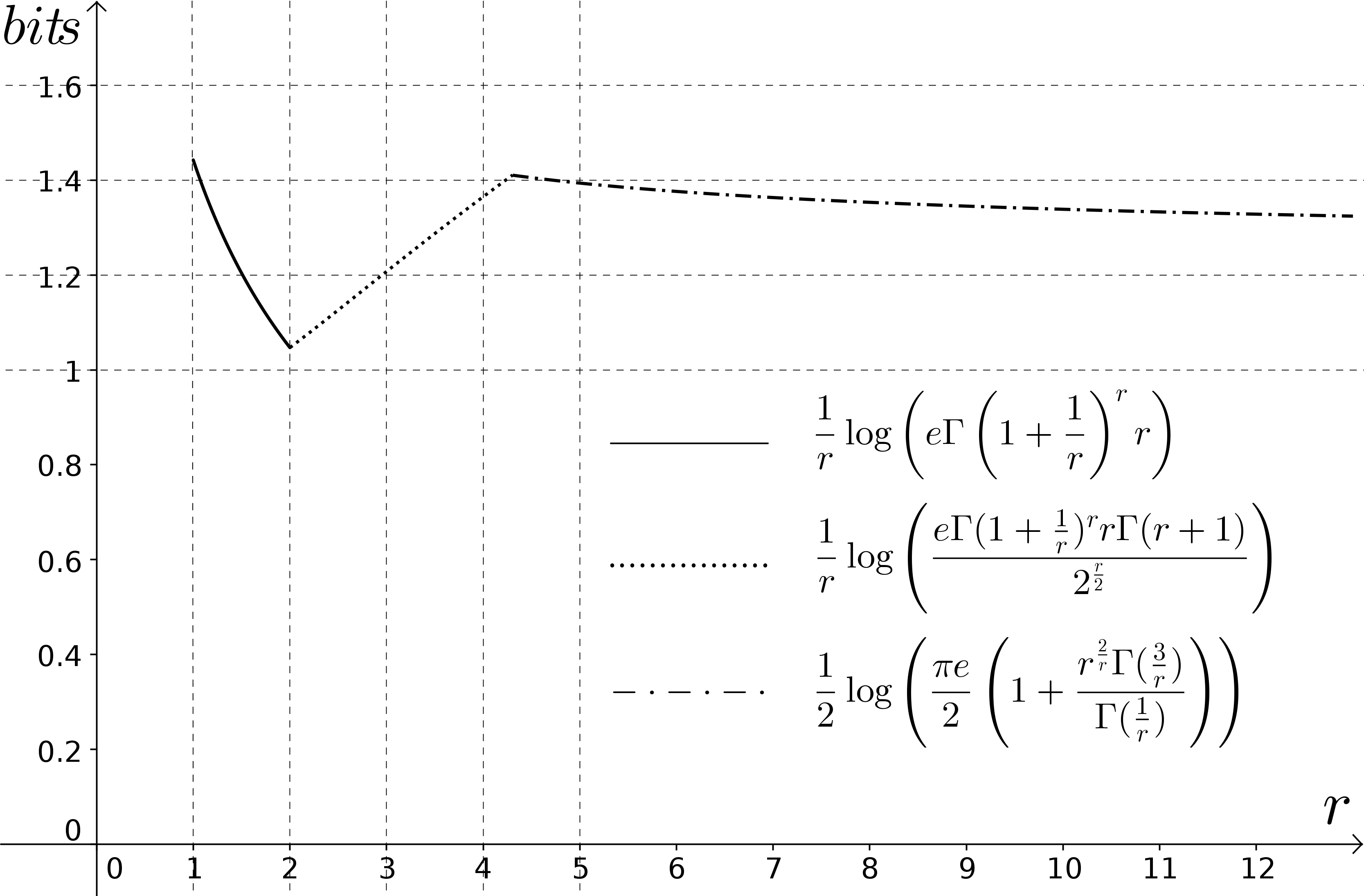}}
\caption{The bound on the difference between the rate-distortion function under $r$-th moment constraint and the Shannon lower bound, stated in Corollary \ref{cons}.}
\label{fig:slb}
\end{figure}

We also provide bounds on the rate-distortion function for vector $X$ (see Theorem \ref{rate-vector} in \secref{sec:main}).

Besides leading to the reverse entropy power inequality and the reverse Shannon lower bound, the new bounds on the differential entropy allow us to bound the capacity of additive noise memoryless channels, provided that the noise follows a log-concave distribution.

The capacity of a  channel that adds a memoryless noise $Z$ is given by (see e.g. \cite[Chapter 9]{cover2012elements}),
\begin{equation}
\mathbb {\mathbb C}_Z(P) = \sup_{X \colon \EE[|X|^2] \leq P} I(X;X+Z), \label{addcap}
\end{equation}
where $P$ is the power allotted for the transmission. 
As a consequence of the entropy power inequality \eqref{EPI}, it holds that
\begin{equation}\label{saddle-point}
{\mathbb C}_Z(P) \geq \ushort {\mathbb C}_Z(P) = \frac{1}{2} \log \left( 1 + \frac{P}{\Var{Z}} \right),
\end{equation}
for arbitrary noise $Z$, where $\ushort {\mathbb C}_Z(P)$ denotes the capacity of the additive white Gaussian noise channel with noise variance $\Var{Z}$. This fact is well known (see e.g. \cite[Chapter 9]{cover2012elements}), and is referred to as the saddle-point condition.

In this paper, we show that whenever the noise $Z$ is log-concave, the difference between the capacity ${\mathbb C}_Z(P)$ and the capacity of a Gaussian channel with the same noise power satisfies
\begin{equation}\label{cap-log}
{\mathbb C}_Z(P) - \ushort {\mathbb C}_Z(P) \leq \log\sqrt{\frac{\pi e}{2}} \approx 1.05 \mbox{ bits}.
\end{equation}

Let us mention a similar result of Zamir and Erez \cite{ZE} who showed that the capacity of an arbitrary memoryless  additive noise  channel is well approximated by the mutual information between the Gaussian input and the output of the channel:
\begin{equation}\label{not2bad}
{\mathbb C}_Z(P) - I(X^*;X^*+Z) \leq \frac{1}{2} \mbox{ bits},
\end{equation}
where $X^*$ is a Gaussian input satisfying the power constraint. The bounds \eqref{cap-log} and \eqref{not2bad} are not directly comparable.

The new bounds on rate-distortion function and the channel capacity can be juxtaposed to yield a converse for joint source-channel coding that is universal in the source and the noise distributions, as long as both are log-concave (see \secref{sec:jscc} below).

The rest of the paper is organized as follows. \secref{sec:main} presents and discusses our main results: the lower bounds on differential entropy in Theorems \ref{variable}, \ref{variable-gen} and \ref{vector}, the reverse entropy power inequalities with explicit constants in Theorems \ref{rev-EPI-gen} and \ref{reverse-vector}, the upper bounds on $\Rd - \Rslb$ in Theorems \ref{rate-gen}, \ref{rate} and \ref{rate-vector}, and the bounds on the capacity of memoryless additive channels in Theorems \ref{cap-variable} and \ref{cap-vector}. The convex geometry tools to prove the bounds on differential entropy and the proofs of Theorems \ref{variable}, \ref{variable-gen} and \ref{vector} are presented in \secref{sec:entropy}. In \secref{convex-measures}, we extend our results to the class of $\gamma$-concave random variables. The reverse entropy power inequalities in Theorems \ref{rev-EPI-gen} and \ref{reverse-vector} are proven in \secref{reverse-EPI}. The bounds on the rate-distortion function in Theorems \ref{rate-gen}, \ref{rate} and \ref{rate-vector} are proven in \secref{sec:rd}. The bounds on the channel capacity in Theorems \ref{cap-variable} and \ref{cap-vector} are proven in \secref{channel-capacity}.

\section{Main results}
\label{sec:main}

\subsection{Lower bounds on the differential entropy}
\label{subsec:bound}

A function $f : \R^n \to [0, + \infty)$ is \emph{log-concave} if $\log f : \R^n \to [-\infty, \infty)$ is a concave function. Equivalently, $f$ is log-concave if for every $\lambda \in [0,1]$ and for every $x,y \in \R^n$, one has
\begin{eqnarray}\label{def-log}
f((1-\lambda)x + \lambda y) \geq f(x)^{1-\lambda} f(y)^{\lambda}.
\end{eqnarray}

We say that a random vector $X$ in $\R^n$ is log-concave if it has a probability density function $f_X$ with respect to Lebesgue measure in $\R^n$ such that $f_X$ is log-concave.

Our first result is a lower bound on the differential entropy of symmetric log-concave random variable in terms of its moments. 

\begin{thm}\label{variable}
Let $X$ be a symmetric log-concave random variable. Then, for every $p > -1$,
\begin{eqnarray}
h(X) \geq \log \frac{2 \|X\|_p }{\Gamma(p+1)^{\frac 1 p}}. \label{eq:varthm}
\end{eqnarray}
Moreover, \eqref{eq:varthm} holds with equality for uniform distribution in the limit $p \to -1$.
\end{thm}

As we will see in Theorem \ref{variable-gen}, for $p=2$, the bound \eqref{eq:varthm} tightens as
\begin{eqnarray}\label{p=2}
h(X) \geq \log(2\|X\|_2).
\end{eqnarray}

The difference between the upper bound in \eqref{hp} and the lower bound in \eqref{eq:varthm} grows as $\log(p)$ as $p \to + \infty$, as $\frac{1}{\sqrt{p}}$ as $p \to 0^+$, and reaches its minimum value of $\log(e) \approx 1.4$ bits at $p=1$.

The next theorem, due to Karlin, Proschan and Barlow \cite{KPB}, shows that the moments of a symmetric log-concave random variable are comparable, and demonstrates that the bound in \thmref{variable} tightens as $p \downarrow -1$. 

\begin{thm}[\!\!\cite{KPB}]\label{2}

Let $X$ be a symmetric log-concave random variable. Then, for every $-1 < p \leq q$,
\begin{eqnarray}
\frac{\|X\|_q}{\Gamma(q+1)^{\frac{1}{q}}} \leq \frac{\|X\|_p}{\Gamma(p+1)^{\frac{1}{p}}}. \label{eq:xqxp}
\end{eqnarray}
Moreover, the Laplace distribution satisfies \eqref{eq:xqxp} with equality.

\end{thm}

Combining Theorem \ref{2} with the well known fact that $\|X\|_p$ is non-decreasing in $p$, we deduce that for every symmetric log-concave random variable $X$, for every $-1<p<q$,
\begin{eqnarray}\label{compare-moment}
\|X\|_p \leq \|X\|_q \leq \frac{\Gamma(q+1)^{\frac{1}{q}}}{\Gamma(p+1)^{\frac{1}{p}}} \|X\|_p.
\end{eqnarray}

Using \thmref{variable} and \eqref{p=2}, we immediately obtain the following upper bound for the relative entropy $D(X||G_X)$ between a symmetric log-concave random variable $X$ and a Gaussian $G_X$ with same variance as that of $X$.

\begin{cor}\label{relative-entr}

Let $X$ be a symmetric log-concave random variable. Then, for every $p>-1$,
\begin{eqnarray}\label{relative-entr-eq}
D(X || G_X) \leq \log\sqrt{\pi e} + \Delta_p,
\end{eqnarray}
where $G_X \sim \mathcal{N}(0, \|X\|_2^2)$, and
\begin{equation}\label{delta-p}
\Delta_p \triangleq
\left\{
\begin{array}{ll}
\log \left( \frac{\Gamma(p+1)^{\frac{1}{p}}}{\sqrt{2}} \frac{\|X\|_2}{\|X\|_p} \right) & \mbox{ $p \neq 2$}\\
- \log \sqrt{2} & \mbox{ $p=2$}
\end{array}
\right..
\end{equation}

\end{cor}

\begin{remark}\label{rem}
 The uniform distribution achieves equality in \eqref{relative-entr-eq} in the limit $p \downarrow -1$. Indeed, if $U$ is uniformly distributed on a symmetric interval, then
\begin{equation}\label{rem-eq}
\Delta_p= \log \frac{\Gamma(p+2)^{\frac{1}{p}}}{\sqrt{6}} \xrightarrow[p \to -1]{} \frac{1}{2} \log \frac{1}{6},
\end{equation}
and so in the limit $p \downarrow -1$ the upper bound in \corref{relative-entr} coincides with the true value of $D(U||G_U)$:
\begin{eqnarray}
D(U||G_U) \leq \frac{1}{2} \log \frac{2 \pi e}{12}  = D(U||G_U).
\end{eqnarray}
Note that $\frac 1 {2\pi e}$ is the normalized second moment (that is, the second moment per dimension of a uniform distribution divided by the volume raised to the power of $\frac 2 n$) of an $n$-dimensional ball, in the limit of large $n$, and $\frac 1 {12}$ is the normalized second moment of a hypercube.
\end{remark}

We next provide a lower bound for the differential entropy of log-concave random variables that are not necessarily symmetric.

\begin{thm}\label{variable-gen}

Let $X$ be a log-concave random variable. Then, for every $p \geq 1$,
\begin{eqnarray}\label{eq:varthm-gen}
h(X) \geq \log \frac{2 \|X - \EE[X]\|_p }{\Gamma(p+1)^{\frac 1 p}}.
\end{eqnarray}
Moreover, for $p=2$, the bound \eqref{eq:varthm-gen} tightens as
\begin{equation}\label{p=2-gen}
h(X) \geq \log(2 \sqrt{\Var{X}}).
\end{equation}

\end{thm}

The next proposition is an analog of Theorem \ref{2} for log-concave random variables that are not necessarily symmetric.

\begin{prop}\label{2-gen}

Let $X$ be a log-concave random variable. Then, for every $1 \leq p \leq q$,
\begin{equation}\label{rev-holder-gen}
\frac{\|X - \EE[X]\|_q}{\Gamma(q+1)^{\frac{1}{q}}} \leq 2 \frac{\|X - \EE[X]\|_p}{\Gamma(p+1)^{\frac{1}{p}}}.
\end{equation}

\end{prop}

\begin{remark}

Contrary to Theorem \ref{2}, we do not know whether there exists a distribution that realizes equality in \eqref{rev-holder-gen}.

\end{remark}

Using \thmref{variable-gen}, we immediately obtain the following upper bound for the relative entropy $D(X||G_X)$ between an arbitrary log-concave random variable $X$ and a Gaussian $G_X$ with same variance as that of $X$. Recall the definition of $\Delta_p$ in \eqref{delta-p}.

\begin{cor}\label{relative-entr-gen}

Let $X$ be a zero-mean, log-concave random variable. Then, for every $p \geq 1$,
\begin{eqnarray}\label{relative-entr-eq-gen}
D(X || G_X) \leq \log\sqrt{\pi e} + \Delta_p,
\end{eqnarray}
where $G_X \sim \mathcal{N}(0, \|X\|_2^2)$. In particular, by taking $p=2$, we necessarily have
\begin{eqnarray}\label{relative-entr-gen-2}
D(X || G_X) \leq \log\sqrt{\frac{\pi e}{2}}.
\end{eqnarray}

\end{cor}

For a given distribution of $X$, one can optimize over $p$ to further tighten \eqref{relative-entr-gen-2}, as seen in \eqref{rem-eq} for the uniform distribution.

We now present a generalization of the bound in \thmref{variable} to random vectors satisfying a symmetry condition. A function $f \colon \R^n \to \R$ is called \emph{unconditional} if for every $(x_1, \dots, x_n) \in \R^n$ and every $(\eps_1, \dots, \eps_n) \in \{-1,1\}^n$, one has
\begin{eqnarray}
f(\eps_1 x_1, \dots, \eps_n x_n) = f(x_1, \dots, x_n).
\end{eqnarray}
For example, the probability density function of the standard Gaussian distribution is unconditional. We say that a random vector $X$ in $\R^n$ is unconditional if it has a probability density function $f_X$ with respect to Lebesgue measure in $\R^n$ such that $f_X$ is unconditional.

\begin{thm}\label{vector}

Let $X$ be a symmetric log-concave random vector in $\R^n$, $n \geq 2$. Then,
\begin{eqnarray}\label{eq:vec}
h(X) \geq \frac{n}{2} \log \frac{|K_X|^{\frac{1}{n}}}{c(n)},
\end{eqnarray}
where $|K_X|$ denotes the determinant of the covariance matrix of $X$, and
\begin{eqnarray}\label{iso-1}
c(n) = \frac{e^2 n^2}{4 \sqrt{2} (n+2)}.
\end{eqnarray}
If, in addition, $X$ is unconditional, then
\begin{eqnarray}\label{iso-2}
c(n) = \frac{e^2}{2}.
\end{eqnarray}

\end{thm}

\begin{remark}

The constant \eqref{iso-2} is better than the constant \eqref{iso-1} if $n \geq 5$.

\end{remark}

By combining Theorem \ref{vector} with the well known upper bound on the differential entropy, we deduce that for every symmetric log-concave random vector $X$ in $\R^n$,
\begin{equation}\label{compare-entr-vect}
\frac{n}{2} \log \left( \frac{|K_X|^{\frac{1}{n}}}{c(n)} \right) \leq h(X) \leq \frac{n}{2} \log \left( 2 \pi e |K_X|^{\frac{1}{n}} \right),
\end{equation}
where $c(n) = \frac{e^2 n^2}{4 \sqrt{2} (n+2)}$ in general, and $c(n) = \frac{e^2}{2}$ if, in addition, $X$ is unconditional.

Using \thmref{vector}, we immediately obtain the following upper bound for the relative entropy $D(X||G_X)$ between a symmetric log-concave random vector $X$ and a Gaussian $G_X$ with same covariance matrix as that of $X$.

\begin{cor}\label{relative-entr-vector}

Let $X$ be a symmetric log-concave random vector in $\R^n$. Then, 
\begin{eqnarray}\label{relative-entr-eq-vector}
D(X || G_X) \leq \frac{n}{2} \log(2 \pi e c(n)),
\end{eqnarray}
where $G_X \sim \mathcal{N}(0, K_X)$, with $c(n) = \frac{n^2 e^2}{(n+2)4 \sqrt{2}}$ in general, and $c(n) = \frac{e^2}{2}$ when $X$ is unconditional.

\end{cor}

\subsection{Extension to $\gamma$-concave random variables}

The bound in \thmref{variable} can be extended to a larger class of random variables than log-concave, namely the class of $\gamma$-concave random variables that we describe next.

Let $\gamma < 0$. We say that a probability density function $f : \R^n \to [0, +\infty)$ is \emph{$\gamma$-concave} if $f^{\gamma}$ is convex. Equivalently, $f$ is $\gamma$-concave if for every $\lambda \in [0,1]$ and every $x,y \in \R^n$, one has
\begin{equation}\label{def-gamma}
f((1-\lambda) x + \lambda y) \geq ((1-\lambda)f(x)^{\gamma} + \lambda f(y)^{\gamma})^{\frac{1}{\gamma}}.
\end{equation}
As $\gamma \to 0$, \eqref{def-gamma} agrees with \eqref{def-log}, and thus $0$-concave distributions corresponds to log-concave distributions. The class of $\gamma$-concave distributions has been deeply studied in \cite{Borell2}, \cite{Borell3}.

Since for fixed $a,b \geq 0$ the function $((1-\lambda)a^{\gamma} + \lambda b^{\gamma})^{\frac{1}{\gamma}}$ is non-decreasing in $\gamma$, we deduce that any log-concave distribution is $\gamma$-concave, for any $\gamma < 0$.

For example, extended Cauchy distributions, that is, distributions of the form
\begin{equation}
f_X(x) = \frac{C_{\gamma}}{1+|x|^{n - \frac{1}{\gamma}}}, \quad x \in \R^n,
\end{equation}
where $C_{\gamma}$ is the normalization constant, are $\gamma$-concave distributions.

We say that a random vector $X$ in $\R^n$ is $\gamma$-concave if it has a probability density function $f_X$ with respect to Lebesgue measure in $\R^n$ such that $f_X$ is $\gamma$-concave.

We derive the following lower bound on the differential entropy for 1-dimensional symmetric $\gamma$-concave random variables, with $\gamma \in (-1,0)$.

\begin{thm}\label{convex-variable}

Let $\gamma \in (-1,0)$. Let $X$ be a symmetric $\gamma$-concave random variable. Then, for every $p \in (-1, -1 - \frac{1}{\gamma})$,
\begin{equation}\label{convex-eq}
h(X) \geq \log \left( \frac{2 \|X\|_p}{\Gamma(p+1)^\frac{1}{p}} \frac{\Gamma(-1 -\frac{1}{\gamma})^{1+\frac{1}{p}}}{\Gamma(-\frac{1}{\gamma}) \Gamma(-\frac{1}{\gamma} - (p+1))^\frac{1}{p}} \right).
\end{equation}

\end{thm}

Notice that \eqref{convex-eq} reduces to \eqref{eq:varthm} as $\gamma \to 0$.

\thmref{convex-variable} implies the following relation between entropy and second moment, for any $\gamma \in (-\frac{1}{3},0)$.

\begin{cor}\label{convex-variable-2}

Let $\gamma \in (-\frac{1}{3},0)$. Let $X$ be a symmetric $\gamma$-concave random variable. Then,
\begin{equation}\label{convex-eq-2}
h(X) \geq \frac{1}{2} \log \left( 2 \|X\|_2^2 \frac{\Gamma(-1 -\frac{1}{\gamma})^{3}}{\Gamma(-\frac{1}{\gamma})^2 \Gamma(-\frac{1}{\gamma} - 3)} \right) = \frac{1}{2} \log \left( 2 \|X\|_2^2 \frac{(2\gamma + 1) (3\gamma + 1)}{(\gamma + 1)^2} \right).
\end{equation}

\end{cor}

\subsection{Reverse entropy power inequality with an explicit constant}

As an application of Theorems \ref{variable-gen} and \ref{vector}, we establish in Theorems \ref{rev-EPI-gen} and \ref{reverse-vector} below a reverse form of the entropy power inequality \eqref{EPI} with explicit constants, for uncorrelated log-concave random vectors. Recall the definition of the entropy power \eqref{entropy-power}.

\begin{thm}\label{rev-EPI-gen}

Let $X$ and $Y$ be uncorrelated log-concave random variables. Then,
\begin{equation}
N(X+Y) \leq \frac{\pi e}{2} (N(X)+N(Y)).
\end{equation} 

\end{thm}

One cannot have a reverse entropy power inequality in higher dimension for arbitrary log-concave random vectors. Indeed, just consider $X$ uniformly distributed on $[-\frac{\eps}{2}, \frac{\eps}{2}] \times [-\frac{1}{2}, \frac{1}{2}]$ and $Y$ uniformly distributed on $[-\frac{1}{2}, \frac{1}{2}] \times [-\frac{\eps}{2}, \frac{\eps}{2}]$ in $\R^2$, with $\eps > 0$ small enough so that $N(X)$ and $N(Y)$ are arbitrarily small compared to $N(X+Y)$. Hence, we need to put $X$ and $Y$ in a certain position so that a reverse form of \eqref{EPI} is possible. While the isotropic position (discussed in Section \ref{sec:entropy}) will work, it can be relaxed to the weaker condition that the covariance matrices are proportionals. Recall that we denote by $K_X$ the covariance matrix of $X$.

\begin{thm}\label{reverse-vector}

Let $X$ and $Y$ be uncorrelated symmetric log-concave random vectors in $\R^n$ such that $K_X$ and $K_Y$ are proportionals. Then,
\begin{equation}\label{eq:rev-EPI-vec}
N(X+Y) \leq \frac{\pi e^3 n^2}{2 \sqrt{2}(n+2)} (N(X)+N(Y)).
\end{equation}
If, in addition, $X$ and $Y$ are unconditional, then
\begin{equation}
N(X+Y) \leq \pi e^3 (N(X)+N(Y)).
\end{equation}

\end{thm}

\subsection{New bounds on the rate-distortion function}

As an application of Theorems \ref{variable} and \ref{variable-gen}, we show in Corollary \ref{cons-gen} and \ref{cons} below that in the class of 1-dimensional log-concave distributions, the rate-distortion function does not exceed the Shannon lower bound by more than $\log(\sqrt{\pi e}) \approx 1.55$ bits (which can be refined to $\log(e) \approx 1.44$ bits when the source is symmetric), independently of $d$ and $r \geq 1$. Denote for brevity
\begin{align}
\beta_r &\triangleq\sqrt{ 1 + \frac{r^{\frac{2}{r}} \Gamma(\frac{3}{r})}{\Gamma(\frac{1}{r})} } \label{betar},
\end{align}
and recall the definition of $\alpha_r$ in \eqref{alphar}.

We start by giving a bound on the difference between the rate-distortion function and the Shannon lower bound, which applies to general, not necessarily log-concave, random variables. 
\begin{thm}\label{rate-gen} 
Let $d \geq 0$ and $r \geq 1$. Let $X$ be an arbitrary random variable. 

$\mathrm{1)}$ Let $r \in [1,2]$. If $\|X\|_2 > d^{\frac{1}{r}}$, then
\begin{eqnarray}
\Rd - \Rslb \leq D(X||G_X) + \log  \frac{\alpha_r}{\sqrt{2 \pi e} } . \label{eq:ub}
\end{eqnarray}
If $\|X\|_2 \leq d^{\frac{1}{r}}$, then $\Rd = 0$. \\

$\mathrm{2)}$ Let $r > 2$. If $\|X\|_2 \geq d^{\frac{1}{r}}$, then
\begin{eqnarray}
\Rd - \Rslb \leq D(X||G_X) + \log \beta_r. \label{eq:ub2}
\end{eqnarray}
If $\|X\|_r \leq d^{\frac 1 r}$, then $\Rd = 0$. \\
If $\|X\|_r > d^{\frac 1 r}$ and $\|X\|_2 < d^{\frac 1 r}$, then
\begin{eqnarray*}
\Rd \leq \log \frac{\sqrt{2 \pi e}\beta_r}{\alpha_r}.
\end{eqnarray*}

\end{thm}

\begin{remark}\label{rem-rate}

For Gaussian $X$ and $r = 2$, the upper bound in \eqref{eq:ub} is $0$, as expected.

\end{remark}

To bound $\Rd - \Rslb$ independently of the distribution of $X$, we apply the bound \eqref{relative-entr-gen-2} on $D(X||G_X)$ to Theorem \ref{rate-gen}:

\begin{cor}\label{cons-gen}

Let $X$ be a log-concave random variable. For $r \in [1,2]$, we have
\begin{eqnarray}
\Rd - \Rslb \leq \log \frac{\alpha_r}{2}.
\end{eqnarray}
For $r>2$, we have
\begin{align}
\Rd - \Rslb \leq \log \left( \sqrt{\frac{\pi e}{2}} \beta_r \right). 
\end{align}

\end{cor}

Please refer to \figref{fig:slb-gen} in \secref{sec:intro} for a numerical evaluation of the bounds of Corollary~\ref{cons-gen}.

The next result refines the bounds in \thmref{rate-gen} for symmetric log-concave random variables when $r>2$.

\begin{thm}\label{rate} 
Let $d \geq 0$ and $r > 2$. Let $X$ be a symmetric log-concave random variable. 

If $\|X\|_2 \geq d^{\frac{1}{r}}$, then
\begin{eqnarray}
\Rd - \Rslb \leq D(X||G_X) + \min \left\{ \log(\beta_r), \log \frac{\alpha_r \Gamma(r+1)^{\frac 1 r}}{2 \sqrt{\pi e} } \right\}. \label{eq:ub3}
\end{eqnarray}
If $\|X\|_r \leq d^{\frac 1 r}$ or $\|X\|_2 \leq \frac{\sqrt 2}{\Gamma(r+1)^{\frac{1}{r}}}d^{\frac{1}{r}}$, then $\Rd = 0$. \\
If $\|X\|_r > d^{\frac 1 r}$ and $\|X\|_2 \in \left(\frac{\sqrt 2}{\Gamma(r+1)^{\frac{1}{r}}}d^{\frac{1}{r}}, d^{\frac{1}{r}}\right)$, then
\begin{equation*}
\Rd \leq \min \left\{ \log  \frac{ \sqrt{2 \pi e}  \beta_r}{ \alpha_r}, \log \frac{\Gamma(r+1)^{\frac 1 r}}{\sqrt 2} \right\}.
\end{equation*}

\end{thm}

\begin{remark}

The bound \eqref{eq:ub3} tightens \eqref{eq:ub2} whenever $r$ is close to $2$.

\end{remark}

By applying the bound \eqref{relative-entr-gen-2} on $D(X||G_X)$ to Theorems \ref{rate-gen} and \ref{rate}, we obtain the following refinement of Corollary \ref{cons-gen} for symmetric log-concave random variables:

\begin{cor}\label{cons}

Let $X$ be a symmetric log-concave random variable. For $r \in [1,2]$, we have
\begin{eqnarray}
\Rd - \Rslb \leq \log \frac{\alpha_r}{2}.
\end{eqnarray}
For $r>2$, we have
\begin{align}
\Rd - \Rslb \leq \min \left\{ \log \frac{\alpha_r \Gamma(r+1)^{\frac 1 r} }{2\sqrt{2}},  ~ \log \left( \sqrt{\frac{\pi e}{2}} \beta_r \right) \right\}. 
\end{align}

\end{cor}

Please refer to \figref{fig:slb} in \secref{sec:intro} for a numerical evaluation of the bounds of Corollary~\ref{cons}. One can see that the graph in \figref{fig:slb} is continuous at $r=2$, contrary to the graph in \figref{fig:slb-gen}. This is because \thmref{2}, which applies to symmetric log-concave random variables, is strong enough to imply the tightening of \eqref{eq:ub2} given in \eqref{eq:ub3}, while \propref{2-gen}, which provides a counterpart of \thmref{2} applicable to all log-concave random variables, is insufficient to derive a similar tightening in that more general setting.

While \corref{cons} bounds the difference $\Rd - \Rslb$ by a universal constant independent of the distribution of $X$, tighter bounds can be obtained  if one is willing to relinquish such universality.
For example, for mean-square distortion ($r = 2$) and a uniformly distributed source $U$, using Remark \ref{rem}, we obtain
\begin{align}
\mathbb R_{U}(d) - \ushort{\mathbb R}_U(d) \leq \frac{1}{2} \log \frac{2 \pi e}{12} \approx 0.254 \mbox{ bits}.
\end{align}

We now turn to the difference between the rate distortion function and the Shannon lower bound for log-concave random vectors in $\R^n$, $n \geq 1$.

For random vector $X$, denote
\begin{equation}
\|X\|_p \triangleq \left(\E{ \sum_{i=1}^n |X_i|^p} \right)^{\frac{1}{p}}.
\end{equation}
The standard rate-distortion function $\Rd$ under $r$-th moment distortion measure is defined as
\begin{eqnarray}
\Rd = \inf_{\substack{ P_{\hat X | {X}} \colon\\ \frac{1}{n}  \|X-\hat{X}\|_r^r \leq d }} I(X;\hat{X}), \label{eq:Rd-vector}
\end{eqnarray}
where the infimum is over all transition probability kernels $\mathbb R^n \mapsto \mathbb R^n$ satisfying the moment constraint. In this context, the Shannon lower bound is
\begin{eqnarray}
\Rd \geq \Rslb \triangleq h(X) - n \log(\alpha_r d^{\frac{1}{r}}),
\end{eqnarray}
where $\alpha_r$ is defined in \eqref{alphar}. The next result, which extends \thmref{rate-gen} to random vectors, not necessarily log-concave, is expressed in terms of the constant $\alpha_r$  and the constant $\beta_r$ defined in \eqref{betar}.

\begin{thm}\label{rate-vector} 

Let $d \geq 0$ and $r \geq 1$. Let $X$ be a random vector in $\R^n$.

$\mathrm{1)}$ let $r \in [1,2]$. If $\|X\|_2 > \sqrt{n} d^{\frac{1}{r}}$, then
\begin{equation}\label{vector:eq1}
\Rd - \Rslb \leq D(X||G_X) + n \log \frac{\alpha_r}{\sqrt{2 \pi e}}  + \frac{n}{2} \log  \frac{\frac{1}{n} \|X\|_2^2}{|K_X|^{\frac{1}{n}}}.
\end{equation}
If $\|X\|_2 \leq \sqrt{n} d^{\frac{1}{r}}$, then $\Rd = 0$. \\

$\mathrm{2)}$ Let $r > 2$. If $\|X\|_2 \geq \sqrt{n} d^{\frac{1}{r}}$, then
\begin{equation}\label{vector:eq2}
\Rd - \Rslb \leq D(X||G_X) + n \log(\beta_r) + \frac{n}{2} \log  \frac{\frac{1}{n} \|X\|_2^2}{|K_X|^{\frac{1}{n}}} .
\end{equation}
If $\|X\|_r \leq (n d)^{\frac{1}{r}}$, then $\Rd = 0$. \\
If $\|X\|_r > (nd)^{\frac 1 r}$ and $\|X\|_2 < \sqrt{n} d^{\frac{1}{r}}$, then
\begin{eqnarray*}
\Rd \leq n \log \frac{\sqrt{2 \pi e}\beta_r}{\alpha_r}.
\end{eqnarray*}

\end{thm}

If $X$ is scalar, the bound in \thmref{rate-vector} reduces to that in \thmref{rate-gen}, in which case the last term in \eqref{vector:eq1} and \eqref{vector:eq2} vanishes.

The bound in Theorem \ref{rate-vector} can be arbitrarily large because of the term $\frac{1}{n} \|X\|^2_2 / |K_X|^{\frac{1}{n}}$. However, for isotropic random vectors (whose definition we recall in Section~\ref{sec:entropy} below), one has $\frac{1}{n} \|X\|^2_2 = |K_X|^{\frac{1}{n}}$. Hence, we deduce the following corollary under mean-square error distortion ($r=2$).

\begin{cor}\label{rate-coro}

Let $X$ be an isotropic random vector in $\R^n$. Then,
\begin{equation}
\Rd - \Rslb \leq D(X||G_X).
\end{equation}

\end{cor}

To bound $\Rd - \Rslb$ independently of the distribution of isotropic log-concave random vector $X$, we apply the bound \eqref{relative-entr-eq-vector} on $D(X||G_X)$ to Corollary \ref{rate-coro}:

\begin{cor}

Let $X$ be an isotropic log-concave random vector in $\R^n$. Then,
\begin{equation}\label{eq:rate-vec-iso}
\Rd - \Rslb \leq \frac{n}{2} \log(2 \pi e \, c(n)),
\end{equation}
where $c(n) = \frac{n^2 e^2}{(n+2)4 \sqrt{2}}$ in general, and $c(n) = \frac{e^2}{2}$ if, in addition, $X$ is unconditional.

\end{cor}

\subsection{New bounds on the capacity of memoryless additive channels}

As another application of Theorem \ref{variable-gen}, we compare the capacity ${\mathbb C}_Z$ of a channel with log-concave additive noise $Z$ with the capacity of the Gaussian channel. Recall that the capacity of the Gaussian channel is
\begin{equation}
\ushort {\mathbb C}_Z(P) = \frac{1}{2} \log \left( 1 + \frac{P}{\Var{Z}} \right).
\end{equation}

\begin{thm}\label{cap-variable}

Let $Z$ be a log-concave random variable. Then,
\begin{equation}\label{cap-variable-eq}
0 \leq {\mathbb C}_Z(P) - \ushort {\mathbb C}_Z(P) \leq  \log\sqrt{\frac{\pi e}{2}}.
\end{equation}

\end{thm}

\begin{remark}

Theorem \ref{cap-variable} tells us that the capacity of a channel with log-concave additive noise exceeds the capacity a Gaussian channel by no more than
\begin{equation}
\log\sqrt{\frac{\pi e}{2}} \approx 1.05 \mathrm{\,\,bits}.
\end{equation}

\end{remark}

As an application of Theorem \ref{vector}, we can provide bounds for the capacity of a channel with log-concave additive noise $Z$ in $\R^n$, $n \geq 1$. The formula for capacity \eqref{addcap} generalizes to dimension $n$ as
\begin{equation}
{\mathbb C}_Z(P) = \sup_{X \colon \frac{1}{n} \|X\|_2^2 \leq P} I(X;X+Z).
\end{equation}

\begin{thm}\label{cap-vector}

Let $Z$ be a symmetric log-concave random vector in $\R^n$. Then,
\begin{equation}
0 \leq {\mathbb C}_Z(P) - \frac{n}{2} \log \left( 1 + \frac{P}{|K_Z|^{\frac{1}{n}}} \right) \leq \frac{n}{2} \log \left( 2 \pi e ~ c(n) \left( \frac{\frac{1}{n} \|Z\|_2^2 + P}{|K_Z|^{\frac{1}{n}} + P} \right) \right),
\end{equation}
where $c(n) = \frac{n^2 e^2}{(n+2)4 \sqrt{2}}$. If, in addition, $Z$ is unconditional, then $c(n) = \frac{e^2}{2}$.

\end{thm}

Similarly to Theorem \ref{rate-vector}, the upper bound in Theorem \ref{cap-vector} can be arbitrarily large by inflating the ratio $\frac{1}{n} \|X\|^2_2 / |K_X|^{\frac{1}{n}}$. For isotropic random vectors (whose definition is recalled in Section \ref{sec:entropy} below), one has $\frac{1}{n} \|Z\|_2^2 = |K_Z|^{\frac{1}{n}}$, and the following corollary follows.

\begin{cor}

Let $Z$ be an isotropic log-concave random vector in $\R^n$. Then,
\begin{equation}\label{eq:cap-vec-iso}
0 \leq {\mathbb C}_Z(P) - \frac{n}{2} \log \left( 1 + \frac{P}{|K_Z|^{\frac{1}{n}}} \right) \leq \frac{n}{2} \log \left( 2 \pi e \, c(n) \right),
\end{equation}
where $c(n) = \frac{n^2 e^2}{(n+2)4 \sqrt{2}}$. If, in addition, $Z$ is unconditional, then $c(n) = \frac{e^2}{2}$.

\end{cor}

\subsection{An explicit converse for joint source-channel coding}
\label{sec:jscc}

In this subsection we illustrate an application of our bounds to joint source-channel coding. 
Consider the general information-theoretic model consisting of a source $S$, whose encoded version, $X$, is transmitted through an additive noise  channel. The decoder observes the output of the channel and forms $\hat S$, its estimate of $S$.

From the data processing inequality (see e.g. \cite{cover2012elements}), we have
\begin{equation}
I(S;\hat{S}) \leq I(X;X+Z).
\end{equation}
Taking the infimum of the left side over all transition probability kernels $P_{\hat{S}|S}$ such that $\EE[|S-\hat{S}|^2] \leq d$, and taking the supremum of the right side over all input $X$ such that $\EE[|X|^2] \leq P$, we obtain the code-independent converse 
\begin{equation}\label{data-1}
\mathbb R_S(d) \leq {\mathbb C}_Z(P).
\end{equation}
Assuming that both $S$ and $Z$ are log-concave, we can further weaken \eqref{data-1} to obtain a universal converse independent of the distributions of $S$ and $Z$. 
From the Shannon lower bound in \eqref{eq:slb2} and Theorem \ref{variable}, we have
\begin{equation}\label{data-2}
\mathbb R_S(d)  \geq \frac{1}{2} \log  \frac{\sigma^2}{d}  - \frac{1}{2} \log \left( \frac{\pi e}{2} \right),
\end{equation}
where $\sigma^2 \triangleq \Var{S}$, 
and from Theorem \ref{cap-variable}, we have
\begin{equation}\label{data-3}
{\mathbb C}_Z(P) \leq \frac{1}{2} \log \left( 1 + \mathrm{snr} \right) + \frac{1}{2} \log \left( \frac{\pi e}{2} \right),
\end{equation}
where we denoted $\mathrm{snr}  \triangleq \frac{P}{\Var{Z}}$.
By combining \eqref{data-1}, \eqref{data-2} and \eqref{data-3}, we deduce that
\begin{equation}
\frac{1}{2} \log  \frac{\sigma^2}{d}  - \frac{1}{2} \log \left( \frac{\pi e}{2} \right) \leq \frac{1}{2} \log \left( 1 + \mathrm{snr} \right) + \frac{1}{2} \log \left( \frac{\pi e}{2} \right).
\end{equation}
Finally, we arrive to
\begin{equation}\label{converse}
d \geq \left(\frac{2}{\pi e} \right)^2 \frac{\sigma^2}{1 + \mathrm{snr} }.
\end{equation}

Notice that the converse bound \eqref{converse} has the same form as the one for the transmission of a Gaussian source over an additive Gaussian channel, which is known to be
\begin{equation}
d \geq \frac{\sigma^2}{1 + \mathrm{snr}}.
\end{equation}

\section{New lower bounds on the differential entropy}
\label{sec:entropy}

In this section, we present the main tools of the paper and prove Theorems \ref{variable}, \ref{variable-gen} and \ref{vector}. The key to our development is the following result for 1-dimensional log-concave distributions, well known in convex geometry. It can be found in \cite{Borell}, in a slightly different form.

\begin{lemma}[{\!\!\cite{Borell}}]
\label{key}

The function
\begin{eqnarray}
F(r) = \frac{1}{\Gamma(r+1)} \int_0^{+ \infty} x^r f(x) \, \de x
\end{eqnarray}
is log-concave on $[-1, +\infty)$, whenever $f \colon [0;+\infty) \to [0;+\infty)$ is log-concave.

\end{lemma}

\begin{proof}[Proof of Theorem \ref{variable}]
Let $p>0$. Applying Lemma \ref{key} to the values $-1,0,p$, we have
\begin{align}
F(0) = F \left( \frac{p}{p+1}(-1) + \frac{1}{p+1} p \right) \geq F(-1)^{\frac{p}{p+1}} F(p)^{\frac{1}{p+1}}. \label{1a}
\end{align}
The bound in \thmref{variable} will follow by computing the values $F(-1)$, $F(0)$ and $F(p)$ for 
$f = f_{X}$.

One has
\begin{align}
F(0) &= \frac 1 2, \label{1b}\\
F(p) &= \frac{\|X\|_p^p}{2 \Gamma(p+1)} .  \label{1c}
\end{align}
To compute $F(-1)$, we first provide a different expression for $F(r)$. Notice that
\begin{align}
 F(r) 
 &=  \frac{1}{\Gamma(r+1)} \int_0^{+\infty} x^r \int_0^{f_{X}(x)} \, \de t \, \de x \\  
 &=  \frac{r+1}{\Gamma(r+2)} \int_0^{\max f_{X}} \int_{\{x \geq 0 \colon f_{X}(x) \geq t\}} x^r \, \de x \, \de t.
\end{align}
Denote the generalized inverse of $f_X$ by
\begin{eqnarray}
f^{-1}_{X}(t) \triangleq \sup\{x \geq 0 \colon f_{X}(x) \geq t \}, \quad t \geq 0.
\end{eqnarray}
Since $f_X$ is log-concave and 
\begin{equation}
 f_X(x) \leq f_X(0) = \max f_X, \label{max0}
\end{equation}
it follows that $f_X$ is non-increasing on $[0, + \infty)$. Therefore,
\begin{equation}
 \{x \geq 0 \colon f_{X}(x) \geq t \} = [0,f_{X}^{-1}(t)].
\end{equation}
Hence,
\begin{align}
F(r) & =  \frac{r+1}{\Gamma(r+2)} \int_0^{ f_{X}(0)} \int_0^{f^{-1}_{X}(t)} x^r \, \de x \, \de t \\ 
& =  \frac{1}{\Gamma(r+2)} \int_0^{f_{X}(0)} (f^{-1}_{X}(t))^{r+1} \, \de t.
\end{align}
We deduce that
\begin{align}
F(-1) = f_{X}(0).  \label{1d}
\end{align}
Plugging \eqref{1b}, \eqref{1c} and \eqref{1d} into \eqref{1a}, we obtain
\begin{eqnarray}\label{moment-max}
f_{X}(0)  \leq  \frac {\Gamma(p+1)^{\frac 1 p}} {2 \|X\|_p} .
\end{eqnarray}
It follows immediately that
\begin{align}
h(X) = \int f_X(x) \log  \frac{1}{f_X(x)} \, \de x \geq \log \frac{1}{f_X(0)} \geq  \log \frac{2 \|X\|_p}{ \Gamma(p+1)^{\frac 1 p}  } . \label{entr}
\end{align}

For $p \in (-1,0)$, the bound is obtained similarly by applying Lemma \ref{key} to the values $-1,p,0$:
\begin{equation}
F(p) = F((-1)(-p) + 0 \cdot (p+1)) \geq F(-1)^{-p} F(0)^{p+1}.
\end{equation}
Hence,
\begin{equation}\label{moment-max-2}
f_X(0) \leq \frac{\Gamma(p+1)^{\frac{1}{p}}}{2 \|X\|_p}.
\end{equation}
We conclude by reproducing \eqref{entr}.

We now show that equality is attained by $U$ uniformly distributed on a symmetric interval $[-\frac{a}{2}, \frac{a}{2}]$, for some $a > 0$. In this case, we have
\begin{equation}
\|U\|_p^p = \left( \frac{a}{2} \right)^p \frac{1}{p+1}.
\end{equation}
Hence,
\begin{equation}
\frac{1}{p} \log \frac{2^p \|U\|_p^p}{\Gamma(p+1)}  = \log  \frac{a}{\Gamma(p+2)^{\frac{1}{p}}} \xrightarrow[p \to -1]{} \log(a) = h(U).
\end{equation}
\end{proof}

\begin{remark}

From \eqref{moment-max}, \eqref{moment-max-2} and \eqref{max0}, we see that the following statement holds: For every symmetric log-concave random variable $X \sim f_X$, for every $p>-1$, and for every $x \in \R$,
\begin{eqnarray}\label{key2}
f_X(x) \leq \frac {\Gamma(p+1)^{\frac 1 p}} {2 \|X\|_p}.
\end{eqnarray}
Inequality \eqref{key2} is the main ingredient in the proof of Theorem \ref{variable}. It is instructive to provide a direct proof of inequality \eqref{key2} without appealing to Lemma \ref{key}, the ideas going back to \cite{KPB}:

\begin{proof}[Proof of inequality \eqref{key2}]
By considering $X | X \geq 0$, where $X$ is symmetric log-concave, it is enough to show that for every log-concave density $f$ supported on $[0,+\infty)$, one has
\begin{equation}\label{101}
f(0) \left( \int_{0}^{+\infty} x^p f(x) \, \de x \right)^{\frac{1}{p}} \leq \Gamma(p+1)^{\frac{1}{p}}.
\end{equation}
By a scaling argument, one may assume that $f(0)=1$. Take $g(x) = e^{-x}$. If $f=g$, then the result follows by a straightforward computation. Assume that $f \neq g$. Since $f \neq g$ and $\int f = \int g$, the function $f-g$ changes sign at least one time. But since $f(0) = g(0)$, $f$ is log-concave and $g$ is log-affine, the function $f-g$ changes sign exactly once. It follows that there exists a unique point $x_0 > 0$ such that for every $0 < x < x_0$, $f(x) \geq g(x)$, and for every $x > x_0$, $f(x) \leq g(x)$. We deduce that for every $x > 0$, and $p \neq 0$,
\begin{equation}
\frac{1}{p}(f(x) - g(x))(x^p - x_0^p) \leq 0.
\end{equation}
Integrating over $x>0$, we arrive at
\begin{align}
\frac{1}{p}\left( \int_{0}^{+\infty} x^p f(x) \, \de x - \Gamma(p+1) \right) & = \frac{1}{p} \int_{0}^{+\infty} x^p (f(x)-g(x)) \, \de x \\ & = \frac{1}{p} \int_{0}^{+\infty} (x^p-x_0^p) (f(x)-g(x)) \, \de x \\\label{105} & \leq 0,
\end{align}
which yields the desired result.
\end{proof}

Actually, the powerful and versatile result of Lemma \ref{key}, which implies \eqref{key2}, is also proved using the technique in \eqref{101}--\eqref{105}. In the context of information theory, Lemma \ref{key} has been previously applied to obtain reverse entropy power inequalities \cite{BM2}, as well as to establish optimal concentration of the information content \cite{fradelizi2015optimal}. In this paper, we make use of Lemma \ref{key} to prove Theorem \ref{variable}. Moreover, Lemma \ref{key} immediately implies \thmref{2}. Below, we recall the argument for completeness.

\end{remark}

\begin{proof}[Proof of \thmref{2}]
The result follows by applying Lemma \ref{key} to the values $0,p,q$. If $0<p<q$, then
\begin{equation}
F(p) = F \left(0 \cdot \left( 1-\frac{p}{q} \right) + q \cdot \left( \frac{p}{q} \right) \right) \geq F(0)^{1-\frac{p}{q}} F(q)^{\frac{p}{q}}.
\end{equation}
Hence,
\begin{equation}
\frac{\|X\|_p^p}{\Gamma(p+1)} \geq \left( \frac{\|X\|_q^q}{\Gamma(q+1)} \right)^{\frac{p}{q}},
\end{equation}
which yields the desired result. The bound is obtained similarly if $p<q<0$ or if $p<0<q$.
\end{proof}

Next, we provide a proof of Theorem \ref{variable-gen} and Proposition \ref{2-gen}, leveraging the ideas from \cite{BM3}.

\begin{proof}[Proof of Theorem \ref{variable-gen}]
Let $Y$ be an independent copy of $X$. Jensen's inequality yields:
\begin{equation}\label{eq:gen}
h(X) = -\int f_X \log(f_X) \geq - \log \left( \int f_X^2 \right) = - \log(f_{X-Y}(0)).
\end{equation}
Since $X-Y$ is symmetric and log-concave, we can apply inequality \eqref{moment-max} to $X-Y$ to obtain
\begin{equation}\label{eq:gen2}
\frac{1}{f_{X-Y}(0)} \geq \frac{2 \|X-Y\|_p}{\Gamma(p+1)^{\frac{1}{p}}} \geq \frac{2 \|X-\EE[X]\|_p}{\Gamma(p+1)^{\frac{1}{p}}},
\end{equation}
where the last inequality again follows from Jensen's inequality. Combining \eqref{eq:gen} and \eqref{eq:gen2} leads to the desired result:
\begin{equation}
h(X) \geq \log \left( \frac{1}{f_{X-Y}(0)} \right) \geq \log \left( \frac{2 \|X-\EE[X]\|_p}{\Gamma(p+1)^{\frac{1}{p}}} \right).
\end{equation}
For $p=2$, one may tighten \eqref{eq:gen2} by noticing that
\begin{equation}
\|X-Y\|^2_2 = 2\Var{X}.
\end{equation}
Hence,
\begin{equation}
h(X) \geq \log \left( \frac{1}{f_{X-Y}(0)} \right) \geq \log \left( \sqrt{2} \|X-Y\|_2 \right) = \log(2 \sqrt{\Var{X}}).
\end{equation}
\end{proof}

\begin{proof}[Proof of Proposition \ref{2-gen}]
Let $Y$ be an independent copy of $X$. Since $X-Y$ is symmetric and log-concave, we can apply Theorem \ref{2} to $X-Y$. Jensen's inequality and triangle inequality yield:
\begin{equation}
\|X-\EE[X]\|_q \leq \|X-Y\|_q \leq \frac{\Gamma(q+1)^{\frac{1}{q}}}{\Gamma(p+1)^{\frac{1}{p}}} \|X-Y\|_p \leq 2 \frac{\Gamma(q+1)^{\frac{1}{q}}}{\Gamma(p+1)^{\frac{1}{p}}} \|X-\EE[X]\|_p.
\end{equation}
\end{proof}

We now study random vectors $X = (X_1, \dots, X_n)$, where $X_i$, $i \geq 1$, may be dependent. We say that a random vector $X \sim f_X$ is \emph{isotropic} if $X$ is symmetric and for all unit vectors $\theta$, one has
\begin{eqnarray}
\EE[\langle X,\theta \rangle ^2] = m^2_{X},
\end{eqnarray}
for some constant $m_{X} > 0$. Equivalently, $X$ is isotropic if its covariance matrix $K_X$ is a multiple of the identity matrix $I_n$, 
\begin{equation}
 K_X = m_{X}^2 I_n,
\end{equation}
for some constant $m_{X} > 0$.
The constant
\begin{eqnarray}
\ell_{X} \triangleq f_X(0)^{\frac{1}{n}} m_{X}
\end{eqnarray}
is called the \emph{isotropic constant} of $X$. 

It is well known that $\ell_{X}$ is bounded from below by a positive constant independent of the dimension \cite{B}. A long-standing conjecture in convex geometry, the \emph{hyperplane conjecture}, asks whether the isotropic constant of an isotropic log-concave random vector is also bounded from above by a universal constant (independent of the dimension). This conjecture holds under additional assumptions, but in full generality, $\ell_{X}$ is known to be bounded only by a constant that depends on the dimension.  For further details, we refer the reader to \cite{BGVV}. We will use the following upper bounds on $\ell_X$ (see \cite{K} for the best dependence on the dimension up to date).

\begin{lemma}\label{iso}

Let $X$ be an isotropic log-concave random vector in $\R^n$, with $n \geq 2$. Then
\begin{eqnarray}
\ell_X^2 \leq \frac{n^2 e^2}{(n+2) 4 \sqrt{2}}.
\end{eqnarray}
If, in addition, $X$ is unconditional, then
\begin{eqnarray}
\ell_X^2 \leq \frac{e^2}{2}.
\end{eqnarray}
If $X$ is uniformly distributed on a convex set, these bounds hold without factor $e^2$.

\end{lemma}

Even though the bounds in \lemref{iso} are well known, we could not find a reference in the literature. We thus include a short proof for completeness.

\begin{proof}
It was shown by Ball \cite[Lemma 8]{B} that if $X$ is uniformly distributed on a convex set, then $\ell_{X}^2 \leq \frac{n^2}{(n+2) 4 \sqrt{2}}$. If $X$ is uniformly distributed on a convex set and is unconditional, then it is known that $\ell_{X}^2 \leq \frac{1}{2}$ (see e.g. \cite[Proposition 2.1]{BN}). Now, one can pass from uniform distributions on a convex set to log-concave distributions at the expense of an extra factor $e^2$, as shown by Ball \cite[Theorem 7]{B}.
\end{proof}

We are now ready to prove Theorem \ref{vector}.

\begin{proof}[Proof of \thmref{vector}]
Let $\widetilde{X} \sim f_{\widetilde{X}}$ be an isotropic log-concave random vector. Notice that $f_{\widetilde{X}}(0)^{\frac{2}{n}}|K_{\widetilde{X}}|^{\frac{1}{n}} = \ell_{\widetilde{X}}^2$, hence, using Lemma \ref{iso}, we have
\begin{eqnarray}\label{iso-eq}
h(\widetilde{X}) = \int f_{\widetilde{X}}(x) \log  \frac{1}{f_{\widetilde{X}}(x) } \de x \geq \log \frac{1}{{f}_{\widetilde{X}}(0)}  \geq \frac{n}{2} \log \frac{|K_{\widetilde{X}}|^{\frac{1}{n}}}{c(n)},
\end{eqnarray}
with $c(n) = \frac{n^2 e^2}{(n+2)4 \sqrt{2}}$. If, in addition, $\widetilde{X}$ is unconditional, then again by Lemma \ref{iso}, $c(n) = \frac{e^2}{2}$.

Now consider an arbitrary symmetric log-concave random vector $X$. One can apply a change of variable to put $X$ in isotropic position. Indeed, by defining $\widetilde{X} = K_X^{-\frac{1}{2}} X$, one has for every unit vector $\theta$,
\begin{align}
\EE[\langle \widetilde{X}, \theta \rangle ^2] = \EE[\langle  X, K_X^{-\frac{1}{2}} \theta \rangle ^2] = \langle K_X (K_X^{-\frac{1}{2}} \theta), K_X^{-\frac{1}{2}} \theta \rangle = 1.
\end{align}
It follows that $\widetilde{X}$ is an isotropic log-concave random vector with isotropic constant $1$. Therefore, we can use \eqref{iso-eq} to obtain
\begin{eqnarray}
h(\widetilde{X}) \geq \frac{n}{2} \log \frac{1}{c(n)},
\end{eqnarray}
where $c(n) = \frac{n^2 e^2}{(n+2)4 \sqrt{2}}$ in general, and $c(n) = \frac{e^2}{2}$ when $X$ is unconditional. We deduce that
\begin{align}
h(X) = h(\widetilde{X}) + \frac{n}{2} \log  |K_X|^{\frac{1}{n}} \geq \frac{n}{2} \log \frac{|K_X|^{\frac{1}{n}}}{c(n)}.
\end{align}
\end{proof}

For isotropic unconditional log-concave random vectors, we extend Theorem \ref{vector} to other moments. First we need the following lemma.

\begin{lemma}\cite[Proposition 3.2]{BN}\label{section}

Let $X \sim f_X$ be an isotropic unconditional log-concave random vector. Then, for every $i \in \{1, \dots, n\}$,
\begin{equation}
f_{X_i}(0) \geq \frac{f_X(0)^{\frac{1}{n}}}{c},
\end{equation}
where $f_{X_i}$ is the marginal distribution of $i$-th component of $X$, i.e. for every $t \in \R$,
\begin{equation}
f_{X_i}(t) = \int_{\R^{n-1}} f_X(x_1, \dots, x_{i-1}, t, x_{i+1}, \dots, x_n) \, \de x_1 \dots \de x_{i-1} \de x_{i+1} \dots \de x_n.
\end{equation}
Here, $c = e \sqrt{6}$. If, in addition, $f_X$ is invariant under permutations of coordinates, then $c = e$.

\end{lemma}

\begin{thm}\label{extend-iso}

Let $X=(X_1, \dots, X_n)$ be an isotropic unconditional log-concave random vector. Then, for every $p>-1$,
\begin{equation}
h(X) \geq \max_{i \in \{1, \dots, n\}} n \log \left( \frac{2 \|X_i\|_p}{\Gamma(p+1)^{\frac{1}{p}}} \frac{1}{c} \right),
\end{equation}
where $c = e\sqrt{6}$. If, in addition, $f_X$ is invariant under permutations of coordinates, then $c=e$.

\end{thm}

\begin{proof}
Let $i \in \{1, \dots, n\}$. We have
\begin{equation}
\|X_i\|_p^p = \int_{\R} |t|^p f_{X_i}(t) \, \de t.
\end{equation}
Since $f_X$ is unconditional and log-concave, it follows that $f_{X_i}$ is symmetric and log-concave, so inequality \eqref{moment-max} applies to $f_{X_i}$:
\begin{equation}
\int_{\R} |t|^p f_{X_i}(t) \, \de t \leq \frac{\Gamma(p+1)}{2^p f_{X_i}(0)^p}. \label{eq:extend-isoa}
\end{equation}
We apply Lemma \ref{section} to pass from $f_{X_i}$ to $f_X$ in the right side of \eqref{eq:extend-isoa}:
\begin{equation}\label{moment-max-iso}
f_X(0)^{\frac{1}{n}} \|X_i\|_p \leq \frac{\Gamma(p+1)^{\frac{1}{p}} c}{2}.
\end{equation}
Thus,
\begin{equation}
h(X) \geq \log  \frac{1}{f_X(0)}  \geq n \log  \frac{2 \|X_i\|_p}{\Gamma(p+1)^{\frac{1}{p}} c}.
\end{equation}
\end{proof}

\section{Extension to $\gamma$-concave random variables}
\label{convex-measures}

In this section, we prove \thmref{convex-variable}, which extends \thmref{variable} to the class of $\gamma$-concave random variables, with $\gamma < 0$. First, we need the following key lemma, which extends Lemma \ref{key}.

\begin{lemma}\label{key-gen}\cite[Theorem 7]{FGP}

Let $f : [0, +\infty) \to [0, +\infty)$ be a $\gamma$-concave function, with $\gamma < 0$. Then, the function
\begin{equation}
F(r) = \frac{\Gamma(-\frac{1}{\gamma})}{\Gamma(-\frac{1}{\gamma} - (r+1))} \frac{1}{\Gamma(r+1)} \int_0^{+\infty} t^r f(t) \, \de t
\end{equation}
is log-concave on $[-1,-1-\frac{1}{\gamma})$.

\end{lemma}

One can recover Lemma \ref{key} from Lemma \ref{key-gen} by letting $\gamma$ go to $0$.

\begin{proof}[Proof of \thmref{convex-variable}]
Let $p \in (-1,0)$. Let us denote by $f_X$ the probability density function of $X$. By applying Lemma \ref{key-gen} to the values $-1,p,0$, we have
$$ F(p) = F(-1 \cdot (-p) + 0 \cdot (p+1)) \geq F(-1)^{-p} F(0)^{p+1}. $$
From the proof of Theorem \ref{variable}, we deduce that $F(-1) = f_X(0)$. Also, notice that
\begin{equation}
F(0) = \frac{1}{2} \frac{\Gamma(-\frac{1}{\gamma})}{\Gamma(-\frac{1}{\gamma} - 1)}.
\end{equation}
Hence,
\begin{equation}
f_X(0)^{-p} \leq \frac{2^p \|X\|_p^p}{\Gamma(p+1)} \frac{\Gamma(-1 -\frac{1}{\gamma})^{p+1}}{\Gamma(-\frac{1}{\gamma})^p \Gamma(-\frac{1}{\gamma}- (p+1))},
\end{equation}
and the bound on differential entropy follows: 
\begin{equation}
h(X) \geq \log  \frac{1}{f_X(0)} \geq \frac{1}{p} \log \left( \frac{2^p \|X\|_p^p}{\Gamma(p+1)} \frac{\Gamma(-1 -\frac{1}{\gamma})^{p+1}}{\Gamma(-\frac{1}{\gamma})^p \Gamma(-\frac{1}{\gamma}- (p+1))} \right).
\end{equation}

If $p \in (0,-1 - \frac{1}{\gamma})$, then the bound is obtained similarly by applying Lemma \ref{key-gen} to the values $-1,0,p$.
\end{proof}

\section{Reverse entropy power inequality with explicit constant}
\label{reverse-EPI}

In this section, we apply the bounds on the differential entropy obtained in Section \ref{sec:main} to derive reverse entropy power inequalities with explicit constants for uncorrelated, possibly dependent random vectors. We start with the 1-dimensional case.

\begin{proof}[Proof of \thmref{rev-EPI-gen}]
Using the upper bound on the differential entropy \eqref{h2}, we have
\begin{equation}
h(X+Y) \leq \frac{1}{2} \log(2 \pi e \Var{X+Y}) = \frac{1}{2} \log(2 \pi e (\Var{X} + \Var{Y})),
\end{equation}
the last equality being valid since $X$ and $Y$ are uncorrelated. Hence,
\begin{equation}
N(X+Y) \leq 2 \pi e (\Var{X} + \Var{Y}).
\end{equation}
Using inequality \eqref{p=2-gen}, we have
\begin{equation}
\Var{X} \leq \frac{N(X)}{4} \quad \mbox{and} \quad \Var{Y} \leq \frac{N(Y)}{4}.
\end{equation}
We conclude that
\begin{equation}
N(X+Y) \leq \frac{\pi e}{2} (N(X) + N(Y)).
\end{equation}
\end{proof}

As a consequence of Corollary \ref{convex-variable-2}, reverse entropy power inequalities for more general distributions can be obtained. In particular, for any uncorrelated symmetric $\gamma$-concave random variables $X$ and $Y$, with $\gamma \in (-\frac{1}{3},0)$,
\begin{equation}
N(X+Y) \leq \pi e \frac{(\gamma + 1)^2}{(2\gamma + 1)(3\gamma + 1)} (N(X) + N(Y)).
\end{equation}

We now prove a reverse entropy power inequality in higher dimension.

\begin{proof}[Proof of \thmref{reverse-vector}]
Since $X$ and $Y$ are uncorrelated and $K_X$ and $K_Y$ are proportionals, 
\begin{equation}
|K_{X+Y}|^{\frac{1}{n}} = |K_X + K_Y|^{\frac{1}{n}} = |K_X|^{\frac{1}{n}} + |K_Y|^{\frac{1}{n}}. \label{eq:reverse-vectora}
\end{equation}
Using \eqref{eq:reverse-vectora} and the upper bound on the differential entropy \eqref{compare-entr-vect}, we obtain
\begin{equation}
h(X+Y) \leq \frac{n}{2} \log \left( 2 \pi e |K_{X+Y}|^{\frac{1}{n}} \right) = \frac{n}{2} \log \left( 2 \pi e \left( |K_X|^{\frac{1}{n}} + |K_Y|^{\frac{1}{n}} \right) \right).
\end{equation}
Using Theorem \ref{vector}, we conclude that
\begin{equation}
N(X+Y) \leq 2 \pi e \left( |K_X|^{\frac{1}{n}} + |K_Y|^{\frac{1}{n}} \right) \leq 2 \pi e \, c(n)(N(X)+N(Y)),
\end{equation}
where $c(n) = \frac{e^2 n^2}{4 \sqrt{2} (n+2)}$ in general, and $c(n)=\frac{e^2}{2}$ if $X$ and $Y$ are unconditional.
\end{proof}

\section{New bounds on the rate-distortion function}
\label{sec:rd}

We are now ready to prove Theorems \ref{rate-gen} and \ref{rate}.

\begin{proof}[Proof of \thmref{rate-gen}]
Denote for brevity  $\sigma = \|X\|_2$.

Under mean-square error distortion ($r=2$), the result is implicit in \cite[Chapter 10]{cover2012elements}. We first recall the argument for the reader convenience, and then proceed to the general case.

Let $r=2$. Assume that $\sigma^2 > d$. We take
\begin{equation}
 \hat{X} = \left(1-\frac{d}{\sigma^2}\right) \left(X+Z \right),
\end{equation}
where $Z \sim \mathcal{N}\left(0,\frac{\sigma^2 d}{\sigma^2 - d} \right)$ is independent of $X$. Note that $\|X - \hat X\|_2^2 = d$. Upper-bounding the rate-distortion function by the mutual information between $X$ and $\hat X$, we obtain
\begin{align}
 \Rd &\leq I(X;\hat{X}) \\
 &= h(\hat{X}) - h(\hat{X} | X) \\
 &= h(X+Z) - h(Z),
\end{align}
where we used homogeneity of differential entropy for the last equality. Invoking the upper bound on the differential entropy \eqref{h2}, we have
\begin{align}
h(X+Z) - h(Z) \leq &~ \frac{1}{2} \log \left( 2 \pi e \left( \sigma^2 + \frac{\sigma^2 d}{\sigma^2 - d} \right) \right) - h(Z) \\ 
 =&~ \frac{1}{2} \log \frac{\sigma^2}{d} \\ 
 =&~ \Rslb + D(X || G_X).
\end{align}
If $\sigma^2 \leq d$, then setting $\hat{X} \equiv 0$ leads to $\Rd = 0$. \\

1) Let $r \in [1,2]$. Assume that $\sigma > d^{\frac{1}{r}}$. We take 
\begin{equation}
 \hat{X} = \left(1-\frac{d^{\frac{2}{r}}}{\sigma^2}\right) \left(X+Z \right), \label{Xhat1}
\end{equation}
where $Z \sim \mathcal{N}\left(0,\frac{\sigma^2 d^{\frac{2}{r}}}{\sigma^2 - d^{\frac{2}{r}}} \right)$ is independent of $X$. This choice of $\hat{X}$ is admissible since
\begin{align}
\|X-\hat{X}\|_r^r \leq&~  \|X-\hat{X}\|_2^r   \\ 
 =&~  \left[ \left( \frac{d^{\frac{2}{r}}}{\sigma^2} \right)^2 \sigma^2 + \left( 1-\frac{d^{\frac{2}{r}}}{\sigma^2} \right)^2 \|Z\|_2^2 \right]^{\frac{r}{2}} \\
 =&~ d,
\end{align}
where we used $r \leq 2$ and the left-hand side of inequality \eqref{compare-moment}.
Upper-bounding the rate-distortion function by the mutual information between $X$ and $\hat X$, we obtain
\begin{align}
 \Rd &\leq I(X;\hat{X})\\
 &= h(\hat{X}) - h(\hat{X} | X) \\
 &= h(X+Z) - h(Z),
\end{align}
where we used homogeneity of entropy for the last equality. Invoking the upper bound on the differential entropy \eqref{h2}, we have
\begin{align}
h(X+Z) - h(Z) \leq &~ \frac{1}{2} \log \left( 2 \pi e \left( \sigma^2 + \frac{\sigma^2 d^{\frac{2}{r}}}{\sigma^2 - d^{\frac{2}{r}}} \right) \right) - h(Z) \\ 
 =&~ \frac{1}{2} \log \frac{\sigma^2}{d^{\frac{2}{r}}} \\ 
 =&~ \Rslb + D(X || G_X) +  \log \frac{\alpha_r}{\sqrt{2 \pi e}},
\end{align}
and \eqref{eq:ub} follows. 

If $\|X\|_2 \leq d^{\frac{1}{r}}$, then $\|X\|_r \leq \|X\|_2 \leq d^{\frac 1 r}$, and setting $\hat{X} \equiv 0$ leads to $\Rd = 0$. \\

2) Let $r > 2$.\footnote{The argument presented here works for every $r \geq 1$. However, for $r \in [1,2]$, the argument in part 1) provides a tighter bound.} Assume that $\sigma \geq d^{\frac{1}{r}}$. We take
\begin{equation}
 \hat{X} = X+Z, \label{Xhatrlarge}
\end{equation}
 where $Z$ is independent of $X$ and realizes the maximum differential entropy under the $r$-th moment constraint, $\|Z\|_r^r=d$. The probability density function of $Z$ is given by
\begin{equation}
f_Z(x) = \frac{r^{1-\frac{1}{r}}}{2 \Gamma \left( \frac{1}{r} \right) d^{\frac{1}{r}}} e^{-\frac{|x|^r}{rd}}, \quad x \in \R.
\end{equation}
Notice that 
\begin{equation}
 \|Z\|_2^2= d^{\frac{2}{r}} \frac{r^{\frac{2}{r}} \Gamma(\frac{3}{r})}{\Gamma(\frac{1}{r})}.
\end{equation}
We have
\begin{align}
h(X+Z) - h(Z) \leq &~ \frac{1}{2} \log(2 \pi e (\sigma^2 + \|Z\|_2^2)) - \log(\alpha_r d^{\frac 1 r}) \label{2d}\\ 
\leq&~ \Rslb + \log \left( \sqrt{2 \pi e} \beta_r \sigma \right) - h(X) \\ 
 =&~ \Rslb + D(X || G_X) +  \log \beta_r,
\end{align}
where $\beta_r$ is defined in \eqref{betar}. 
Hence,
\begin{equation}
\Rd - \Rslb \leq D(X || G_X) + \log \beta_r. \label{rlarge}
\end{equation}

If $\|X\|_r^r \leq d$, then setting $\hat{X} \equiv 0$ leads to $\Rd = 0$.

Finally, if $\|X\|_r^r > d$ and $\sigma < d^{\frac{1}{r}}$, then from \eqref{2d} we obtain
\begin{align}
\Rd &\leq \log \left( \sqrt{2 \pi e} \beta_r d^{\frac{1}{r}} \right) - \log(\alpha_r d^{\frac 1 r}) \\
&= \log  \frac{ \sqrt{2 \pi e}  \beta_r}{ \alpha_r}.
\end{align}
\end{proof}

We now provide a different upper bound for symmetric log-concave random variable that improves \eqref{rlarge} for values of $r>2$ close to 2 (but is worse than \eqref{rlarge} when $r$ is large, see \figref{fig:slb}).

\begin{proof}[Proof of Theorem \ref{rate}]
Denote for brevity  $\sigma = \|X\|_2$, and recall that $X$ is a symmetric log-concave random variable.

Assume that $\sigma \geq d^{\frac{1}{r}}$. We take 
\begin{align}
 \hat{X} &= \left(1-\frac{\delta}{\sigma^2}\right)(X+Z),\label{Xhat2}\\
 \delta &\triangleq \frac{2}{\Gamma(r+1)^{\frac{2}{r}}} d^{\frac{2}{r}},
\end{align}
where $Z \sim \mathcal{N}\left(0,\frac{\sigma^2 \delta}{\sigma^2 - \delta}\right)$ is independent of $X$. This choice of $\hat{X}$ is admissible since
\begin{align}
 \|X-\hat{X}\|_r^r &\leq \|X-\hat{X}\|_2^r \frac{\Gamma(r+1)}{2^{\frac{r}{2}}}\\
  &= \delta^{\frac{r}{2}} \frac{\Gamma(r+1)}{2^{\frac{r}{2}}} \\
  &= d,
\end{align}
where we used $r> 2$ and \thmref{2}.
 Using the upper bound on the differential entropy \eqref{h2}, we have
\begin{align}
 h(X+Z) - h(Z) \leq &~ \frac{1}{2} \log \left( 2\pi e \left( \sigma^2+\frac{\sigma^2 \delta}{\sigma^2-\delta} \right) \right) - h(Z) \\ 
  =&~ \frac{1}{2} \log  \frac{\sigma^2}{\delta}  \label{2c}\\ 
  =&~ \Rslb + D(X || G_X) + \log  \frac{\alpha_r \Gamma(r+1)^{\frac 1 r}}{2 \sqrt{\pi e}}.
\end{align}
Hence,
\begin{equation}\label{r>2}
\Rd - \Rslb \leq D(X || G_X) + \log \frac{\alpha_r \Gamma(r+1)^{\frac 1 r}}{2 \sqrt{\pi e}} .
\end{equation}

If $\sigma^2 \leq \delta$, then from \thmref{2} $\|X\|_r^r \leq d$, hence $\Rd = 0$.

Finally, if $\|X\|_r^r > d$ and $\sigma^2 \in (\delta, d^{\frac{2}{r}})$, then from \eqref{2c} we obtain
\begin{align}
 \Rd &\leq \frac{1}{2} \log  \frac{\sigma^2}{\delta}  \\
 &\leq \frac{1}{2} \log \frac{\Gamma(r+1)^{\frac{2}{r}}}{2}.
\end{align}
\end{proof}

\begin{remark} ~\\

1) Let us explain the strategy in the proof of Theorems \ref{rate-gen} and \ref{rate}. By definition, $\Rd \leq I(X;\hat{X})$ for any $\hat{X}$ satisfying the constraint. In our study, we chose $\hat{X}$ of the form $\lambda(X+Z)$, with $\lambda \in [0,1]$, where $Z$ is independent of $X$. To find the best bounds possible with this choice of $\hat{X}$, we need to minimize $\|X-\hat{X}\|_r^r$ over $\lambda$. Notice that if 
\begin{equation}
 \hat{X}=\lambda(X+Z)
\end{equation}
 and $Z$ symmetric, then 
\begin{equation}
 \|X-\hat{X}\|_r^r = \|(1-\lambda) X + \lambda Z\|_r^r.
\end{equation}
To estimate $\|(1-\lambda) X + \lambda Z\|_r^r$ in terms of $\|X\|_r$ and $\|Z\|_r$, one can use triangle inequality and the convexity of $\|\cdot\|^r$ to get the bound
\begin{equation}\label{choice1}
\|(1-\lambda) X + \lambda Z\|_r^r \leq 2^{r-1} ( (1-\lambda)^r \|X\|_r^r + \lambda^r \|Z\|_r^r),
\end{equation}
or, one can apply Jensen's inequality directly to get the bound
\begin{equation}\label{choice2}
\|(1-\lambda) X + \lambda Z\|_r^r  \leq (1-\lambda) \|X\|_r^r + \lambda \|Z\|_r^r.
\end{equation}
Consider the function
\begin{equation}
F(\lambda) = 2^{r-1} ((1-\lambda)^r a + \lambda^r b)
\end{equation}
that appears on the right-hand side of \eqref{choice1}, and the function
\begin{equation}
G(\lambda) = (1-\lambda) a + \lambda b
\end{equation}
that appears on the right-hand side of \eqref{choice2}, for fixed $a,b>0$. Notice that
\begin{eqnarray}
F(\lambda) = 2^{r-1} ((1-\lambda)^r a + \lambda^r b) & \geq & 2^r \left( \frac{1}{2}(1-\lambda)^r + \frac{1}{2} \lambda^r \right) \min(a,b) \\\label{184} & \geq & \min(a,b) \\ & = & \min_{\lambda \in [0,1]} G(\lambda),
\end{eqnarray}
where in \eqref{184} we use the fact that $\lambda \to \lambda^r$ is convex. Hence,
\begin{equation}
\min_{\lambda \in [0,1]} G(\lambda) \leq \min_{\lambda \in [0,1]} F(\lambda).
\end{equation}
Furthermore, provided that we are in the nontrivial regime $\|X\|_r > \|Z\|_r$, the minimum of $G$ is attained at $\lambda^* = 1$. This justifies choosing $\hat X$ as in  \eqref{Xhatrlarge} in the proof of \eqref{eq:ub2}. 

To justify the choice of $\hat X$ in \eqref{Xhat1} (also in \eqref{Xhat2}), which leads the tightening of \eqref{eq:ub2} for $r \in [1, 2]$ in \eqref{eq:ub} (also in \eqref{eq:ub3}), we bound $r$-th norm by second norm, and we note that by the independence of $X$ and $Z$,
\begin{equation}\label{choice3}
\|(1-\lambda) X + \lambda Z\|_2^2 \leq (1-\lambda)^2 \|X\|_2^2 + \lambda^2 \|Z\|_2^2.
\end{equation}
The function $H(\lambda) = (1-\lambda)^2 a + \lambda^2 b$ that appears on the right-hand side of \eqref{choice3} attains its minimum at $\lambda^* = \frac{a}{a+b}$. The minimum of $H$ is $\frac{ab}{a+b}$, which is always smaller than $\min(a,b)$. That is why we chose $\hat X$ as in \eqref{Xhat1} for $r \in [1, 2]$, and we chose $\hat X$ as in \eqref{Xhat2} for $r>2$. \\

2) Using Corollary \ref{relative-entr-gen}, if $r = 2$ one may rewrite our bound in terms of the rate distortion function of a Gaussian source as follows: 
\begin{align}\label{Rgauss}
\Rd \geq \mathbb R_{G_X}(d) - \log \sqrt{\pi e} - \Delta_p,
\end{align}
where $\Delta_p$ is defined in \eqref{delta-p}, and where
\begin{equation}
\mathbb R_{G_X}(d) = \frac{1}{2} \log  \frac{\sigma^2}{d} 
\end{equation}
is the rate distortion function of a Gaussian source with the same variance $\sigma^2$ as $X$. It is well known that for arbitrary source and mean-square distortion (see e.g. \cite[Chapter 10]{cover2012elements})
\begin{equation}
 \Rd \leq \mathbb R_{G_X}(d). \label{Rgu}
\end{equation}
By taking $p=2$ in \eqref{Rgauss}, we obtain
\begin{equation}
0 \leq \mathbb R_{G_X}(d) - \Rd \leq \frac{1}{2} \log \left( \frac{\pi e}{2} \right). \label{Rgl}
\end{equation}
The bounds in \eqref{Rgu} and \eqref{Rgl} tell us that the rate distortion function of any log-concave source is approximated by that of a Gaussian source. In particular, approximating $\Rd$ of an arbitrary log-concave source by 
\begin{equation}
 \hat{\mathbb R}_X(d) = \frac{1}{2} \log \frac{\sigma^2}{d} - \frac{1}{4} \log \left( \frac{\pi e}{2} \right),
\end{equation}
we guarantee the approximation error $|\Rd - \hat{\mathbb R}_X(d)|$ of at most $\frac{1}{4} \log \left( \frac{\pi e}{2} \right) \approx \frac{1}{2}$ bits.

\end{remark}

\vspace*{0.5cm}

We now establish bounds for the difference between the rate distortion function and the Shannon lower bound for log-concave random vectors in $\R^n$, $n \geq 1$. The proof of Theorem \ref{rate-vector} follows the same lines as the proof of Theorem \ref{rate-gen}, we thus only detail the most interesting case of mean-square error distortion, which reads as follows:

\begin{cor}\label{rate-vector-2}

Let $X$ be a random vector in $\R^n$. If $\frac{1}{n} \|X\|_2^2 > d$, then
\begin{equation}
\Rd - \Rslb \leq D(X||G_X) + \frac{n}{2} \log \left( \frac{\frac{1}{n} \|X\|_2^2}{|K_X|^{\frac{1}{n}}} \right).
\end{equation}

If $\frac{1}{n} \|X\|_2^2 \leq d$, then $\Rd = 0$.

\end{cor}

\begin{proof}
Denote for brevity $\sigma^2 = \frac{1}{n} \|X\|_2^2$.

If $\sigma^2 > d$, then we choose $\hat{X} = (1-\frac{d}{\sigma^2})(X+Z)$, where $Z \sim \mathcal{N}(0, \frac{\sigma^2 d}{\sigma^2 - d} \cdot I_n)$ is independent of $X$. This choice is admissible since
\begin{equation}
\frac{1}{n} \|X-\hat{X}\|^2_2 = \frac{1}{n} \left( \frac{d^2}{\sigma^4} n \sigma^2 + \left( 1 - \frac{d}{\sigma^2} \right)^2 \frac{n \sigma^2 d}{\sigma^2 - d} \right) = d.
\end{equation}
Thus,
\begin{equation}
\Rd \leq h(X+Z) - h(Z) \leq \frac{n}{2} \log \left( 2 \pi e \frac{\|X+Z\|_2^2}{n} \right) - h(Z) = \frac{n}{2} \log \frac{\sigma^2}{d}.
\end{equation}
Hence,
\begin{equation}
\Rd - \Rslb \leq \frac{n}{2} \log(2 \pi e \sigma^2) - h(X) = D(X||G_X) + \frac{n}{2} \log \frac{\sigma^2}{|K_X|^{\frac{1}{n}}} .
\end{equation}

If $\sigma^2 \leq d$, then we can choose $\hat{X} \equiv 0$. Hence, $\Rd = 0$.
\end{proof}

Let us consider the rate distortion function under covariance matrix constraint:
\begin{equation}
\R_X^{cov}(d) = \inf_{P_{\hat{X}|X} \colon |K_{X-\hat{X}}|^\frac{1}{n} \leq d} I(X;\hat{X}), \label{eq:Rdcov}
\end{equation}
where the infimum is taken over all joint distributions satisfying the determinant constraint $|K_{X-\hat{X}}|^\frac{1}{n} \leq d$. For this distortion measure, we have the following bound.

\begin{thm}\label{rate-vector-cov}

Let $X$ be a symmetric log-concave random vector in $\R^n$. If $|K_X|^{\frac{1}{n}} > d$, then
\begin{equation}
0 \leq \R_X^{cov}(d) - \Rslb \leq D(X||G_X) \leq \frac{n}{2} \log(2 \pi e \, c(n)),
\end{equation}
with $c(n) = \frac{n^2 e^2}{(n+2)4 \sqrt{2}}$. If, in addition, $X$ is unconditional, then $c(n) = \frac{e^2}{2}$.

If $|K_X|^{\frac{1}{n}} \leq d$, then $\R_X^{cov}(d) = 0$.

\end{thm}

\begin{proof}
If $|K_X|^{\frac{1}{n}} > d$, then we choose $\hat{X} = \left( 1-\frac{d}{|K_X|^{\frac{1}{n}}} \right)(X+Z)$, where $Z \sim \mathcal{N}\left( 0, \frac{d}{|K_X|^{\frac{1}{n}}-d} \cdot K_X\right)$ is independent of $X$. This choice is admissible since by independence of $X$ and $Z$,
\begin{equation}
K_{X-\hat{X}} = \left( \frac{d}{|K_X|^{\frac{1}{n}}} \right)^2 K_X + \left( 1 -\frac{d}{|K_X|^{\frac{1}{n}}} \right)^2 K_Z,
\end{equation}
hence, using the fact that $K_X$ and $K_Z$ are proportionals,
\begin{equation}
|K_{X-\hat{X}}|^{\frac{1}{n}} = \left( \frac{d}{|K_X|^{\frac{1}{n}}} \right)^2 |K_X|^{\frac{1}{n}} + \left( 1 - \frac{d}{|K_X|^{\frac{1}{n}}} \right)^2 |K_Z|^{\frac{1}{n}} = d.
\end{equation}
Upper-bounding the rate-distortion function by the mutual information between $X$ and $\hat X$, we have
\begin{align}
\R_X^{cov}(d) &\leq h(X+Z) - h(Z) \\
&\leq \frac{n}{2} \log \frac{|K_X|^{\frac{1}{n}}}{d} . \label{eq:rcovu}
\end{align}
Since the Shannon lower bound for covariance constraint coincides with that for the mean-square error constraint, 
\begin{equation}
\R_X^{cov}(d) \geq \Rslb  = h(X) - \frac n 2 \log (2 \pi e d).
\end{equation}
On the other hand, using \eqref{eq:rcovu}, we have
\begin{align}
\R_X^{cov}(d) - \Rslb &\leq \frac{n}{2} \log(2 \pi e |K_X|^{\frac{1}{n}}) - h(X) \\
&= D(X||G_X) \\
&\leq \frac{n}{2} \log(2 \pi e c(n)), \label{eq:cova}
\end{align}
where \eqref{eq:cova} follows from Corollary \ref{relative-entr-vector}.

If $|K_X|^{\frac{1}{n}} \leq d$, then we put $\hat{X} \equiv 0$, which leads to $\R_X^{cov}(d) = 0$.
\end{proof}

\section{New bounds on the capacity of memoryless additive channels}
\label{channel-capacity}

In this section, we study the capacity of memoryless additive channels when the noise is a log-concave random vector. Recall that the capacity of such a channel is
\begin{equation}
{\mathbb C}_Z(P) = \sup_{X \colon \frac{1}{n} \|X\|_2^2 \leq P} I(X;X+Z). \label{eq:cz}
\end{equation}

Notice that,
\begin{align}
I(X;X+Z) &= h(X+Z)-h(X+Z | X) \\
&= h(X+Z)-h(Z | X) \\
&= h(X+Z)-h(Z).
\end{align}
Hence, \eqref{eq:cz} can be re-written as
\begin{equation}
{\mathbb C}_Z(P) = \sup_{X \colon \frac{1}{n} \|X\|_2^2 \leq P} h(X+Z)-h(Z).
\end{equation}

We compare the capacity ${\mathbb C}_Z$ of a channel with log-concave additive noise with the capacity of the Gaussian channel.

\begin{proof}[Proof of \thmref{cap-variable}]
The lower bound is a consequence of the entropy power inequality, as mentioned in \eqref{saddle-point}. To obtain the upper bound, we first use the upper bound on the differential entropy \eqref{h2} to conclude that
\begin{equation}\label{cap-variable-1}
h(X+Z) \leq \frac{1}{2} \log(2 \pi e (P + \Var{Z})),
\end{equation}
for every random variable $X$ such that $\|X\|_2^2 \leq P$. By combining \eqref{cap-variable-1} and \eqref{p=2-gen}, we deduce that
\begin{equation}
\begin{aligned}
{\mathbb C}_Z(P) &= \sup_{X \colon \|X\|_2^2 \leq P} h(X+Z)-h(Z) \\
& \leq \frac{1}{2} \log(2 \pi e (P + \Var{Z})) - \frac{1}{2} \log(4 \Var{Z}) \\ 
& = \frac{1}{2} \log \left( \frac{\pi e}{2} \left( 1 + \frac{P}{\Var{Z}} \right) \right),
\end{aligned}
\end{equation}
which is the desired result.
\end{proof}

\begin{proof}[Proof of Theorem \ref{cap-vector}]
The lower bound is a consequence of the entropy power inequality, as mentioned in \eqref{saddle-point}. To obtain the upper bound, we write
\begin{eqnarray}
h(X+Z) - h(Z) & \leq & \frac{n}{2} \log \left( 2 \pi e |K_{X+Z}|^{\frac{1}{n}} \right) - h(Z) \label{cap1} \\ & \leq & \frac{n}{2} \log \left( 2 \pi e \frac{\|X\|_2^2 + \|Z\|_2^2}{n} \right) - h(Z) \label{cap2} \\ & \leq & \frac{n}{2} \log \left( 2 \pi e \frac{\|X\|_2^2 + \|Z\|_2^2}{n} \right) - \frac{n}{2} \log  \frac{|K_Z|^{\frac{1}{n}}}{c(n)}  \label{cap3} \\ & \leq & \frac{n}{2} \log \left( 2 \pi e ~ c(n) \left( \frac{\frac{1}{n} \|Z\|_2^2}{|K_Z|^{\frac{1}{n}}} + \frac{P}{|K_Z|^{\frac{1}{n}}} \right) \right),
\end{eqnarray}
where $c(n) = \frac{n^2 e^2}{(n+2)4 \sqrt{2}}$ in general, and $c(n) = \frac{e^2}{2}$ if $Z$ is unconditional. Inequality \eqref{cap1} is obtained from the upper bound on the differential entropy \eqref{compare-entr-vect}. Inequality \eqref{cap2} is a consequence of the arithmetic-geometric mean inequality. Finally, inequality \eqref{cap3} is obtained from Theorem \ref{vector}.
\end{proof}

\section{Conclusions}

Several recent results show that the entropy of log-concave probability densities have nice properties. For example, reverse entropy power inequalities, strengthened and stable versions of the entropy power inequality were recently obtained for log-concave random vectors (see e.g. \cite{BM}, \cite{BNG}, \cite{BNT}, \cite{T1}, \cite{T2}, \cite{CFP}). This line of developments suggest that, in some sense, log-concave random vectors behave like Gaussians.
 
Our work follows this line of results, by establishing a new lower bound on differential entropy for log-concave random variables in \eqref{eq:var-gen}, for log-concave random vectors with possibly dependent coordinates in \eqref{eq:vec}, and for $\gamma$-concave random variables in \eqref{convex-eq}. We made use of the new lower bounds in several applications. First, we derived reverse entropy power inequalities with explicit constants for uncorrelated, possibly dependent log-concave random vectors in \eqref{eq:rev-EPI} and \eqref{eq:rev-EPI-vec}. We also showed a universal bound on the difference between the rate-distortion function and the Shannon lower bound for log-concave random variables in \figref{fig:slb-gen} and \figref{fig:slb}, and for log-concave random vectors in \eqref{eq:rate-vec-iso}. Finally, we established an upper bound on the capacity of memoryless additive noise channels when the noise is a log-concave random vector in \eqref{cap-log} and \eqref{eq:cap-vec-iso}.

Under the Gaussian assumption, information-theoretic limits in many communication scenarios admit simple closed-form expressions. Our work demonstrates that at least in three such scenarios (source coding, channel coding and joint source-channel coding), the information-theoretic limits admit closed-form approximations if the Gaussian assumption is relaxed to the log-concave one. We hope that the approach will be useful in gaining insights into those communication and data processing scenarios in which the Gaussianity of the observed distributions is violated but the log-concavity is preserved.

\section*{Acknowledgment.}

The authors would like to thank the anonymous referee for pointing out that the bound \eqref{eq:varthm} and, up to a factor $2$, the bound \eqref{eq:xqxp} also apply to the non-symmetric case if $p \geq 1$ (Theorem \ref{variable-gen}, Proposition \ref{2-gen}).

\bibliographystyle{IEEEtran}
\bibliography{logconcave}

\end{document}